\documentclass[11pt]{article}

\usepackage{amsthm}
\usepackage{amsmath}
\usepackage{amssymb}
\usepackage{amsfonts}
\usepackage{graphicx}
\usepackage{color}
\usepackage{enumerate}
\usepackage{amscd}

\usepackage{float}
\usepackage{booktabs}
\usepackage{multirow}
\usepackage{makecell}

\usepackage{algorithm}
\usepackage{algorithmicx}
\usepackage{algpseudocode}

\usepackage{natbib}

\usepackage{tikz-cd}
\usepackage{url}
\usepackage{wrapfig}
\usepackage{mathrsfs}
\usepackage{mdframed}

\usepackage{dsfont} 

\usepackage{geometry}
\usepackage{setspace}

\newcommand{\R}{\mathbb{R}}
\newcommand{\E}{\mathbb{E}}

\newcommand{\SHR}{\text{\it SHR}}

\newcommand{\T}{\mathbf{T}}
\newcommand{\V}{\mathbf{V}}
\newcommand{\W}{\mathbf{W}}
\newcommand{\X}{\mathbf{X}}
\newcommand{\Y}{\mathbf{Y}}
\newcommand{\bH}{\mathbf{H}}
\newcommand{\h}{\mathbf{h}}
\newcommand{\x}{\mathbf{x}}
\newcommand{\y}{\mathbf{y}}
\newcommand{\z}{\mathbf{z}}
\newcommand{\w}{\mathbf{w}}

\newcommand{\dx}[1]{\ \text{d} #1}
\newcommand{\indicator}[1]{\mathds{1}\!\left\{ #1 \right\}}

\newtheoremstyle{mytheoremstyle}
  {\topsep} 
  {0pt} 
  {} 
  {} 
  {\bfseries} 
  {.} 
  {.5em} 
  {} 
\theoremstyle{mytheoremstyle}

\newtheorem{prop}{Proposition}
\newtheorem{cor}{Corollary}
\newtheorem{lem}{Lemma}
\newtheorem{defn}{Definition}

\allowdisplaybreaks

\newcommand{\blind}{0}

\title{Randomization for the susceptibility effect \\ of an infectious disease intervention}

\if1\blind
{
  \author{[Blinded for review]}
} \fi

\if0\blind
{
\author{Daniel J. Eck$^1$, Olga Morozova$^{2}$, and Forrest W. Crawford$^{2,3,4,5}$ \\[1em]
1. Department of Statistics, University of Illinois Urbana-Champaign \\
2. Department of Biostatistics, Yale School of Public Health \\ 
3. Department of Statistics \& Data Science, Yale University \\
4. Department of Ecology \& Evolutionary Biology, Yale University \\
5. Yale School of Management }

} \fi

\begin{document}

\maketitle

\begin{abstract}
\noindent Randomized trials of infectious disease interventions, such as vaccines, often focus on groups of connected or potentially interacting individuals. When the pathogen of interest is transmissible between study subjects, interference may occur: individual infection outcomes may depend on treatments received by others. Epidemiologists have defined the primary causal effect of interest -- called the ``susceptibility effect'' -- as a contrast in infection risk under treatment versus no treatment, while holding exposure to infectiousness constant.  
A related quantity -- the ``direct effect'' --  is defined as an unconditional contrast between the infection risk under treatment versus no treatment.  
The purpose of this paper is to show that under a widely recommended randomization design, the direct effect may fail to recover the sign of the true susceptibility effect of the intervention in a randomized trial when outcomes are contagious. 
The analytical approach uses structural features of infectious disease transmission to define the susceptibility effect. A new probabilistic coupling argument reveals stochastic dominance relations between potential infection outcomes under different treatment allocations. 
The results suggest that estimating the direct effect under randomization may provide misleading inferences about the effect of an intervention -- such as a vaccine -- when outcomes are contagious.  
\\[1em]
\noindent \textbf{Keywords}: contagion, direct effect, interference, probabilistic coupling, transmission model, vaccine 
\end{abstract}



\section{Introduction}

Randomized trials are widely used in the evaluation of infectious disease interventions among potentially interacting individuals \citep{halloran1997study,datta1999efficiency,halloran2010design}. 
For example, randomized trials have been employed to evaluate the effects of interventions, including vaccines, to prevent  
influenza \citep{belshe1998efficacy,hayden2000inhaled,welliver2001effectiveness,monto2002zanamivir}, 
pertussis \citep{simondon1997randomized}, 
typhoid \citep{acosta2005multi},
and
cholera \citep{clemens1986field,perez2014assessing}, among many other diseases. The primary goal of most infectious disease intervention trials is to estimate the causal effect of treatment on the infection risk of the individual who receives it.  However, when the infection is transmissible, or contagious, between study subjects, the treatment delivered to one subject may affect the infection outcome of others, via prevention of the original subject's infection or reduction in their infectiousness once infected \citep{halloran1991study,halloran1995causal}.  This phenomenon -- called ``interference'' in the causal inference literature -- complicates definition and estimation of causal intervention effects under contagion \citep{halloran1991study,halloran1995causal,vanderweele2011effect,halloran2016dependent,halloran2017simulations,ogburn2017vaccines,ogburn2018challenges}.


The ``susceptibility effect'' is of primary epidemiological interest in vaccine trials because it summarizes the effect of the intervention on the person who receives it, holding exposure to infection constant \citep{halloran1995causal,halloran1997study,golm1999semiparametric,ohagan2014estimating}.  \citet[][page 19]{halloran2010design} write, ``Historically, the primary focus has been how well vaccination protects the vaccinated individual. $VE_S$, the vaccine effiacy for susceptibility, is a measure of how protective vaccination is against infection''.   The susceptibility effect is sometimes called the ``vaccine effect on susceptibility'', the ``conditional direct causal effect'' \citep{halloran1995causal}, or per-exposure effect \citep{ohagan2014estimating}, and may be represented by a hazard ratio, risk ratio, or risk difference \citep{halloran1991direct,halloran1997study,halloran1999design,ohagan2014estimating}.  Unfortunately, the susceptibility effect can be difficult to estimate because exposure to infection cannot always be precisely measured.


A related quantity, called the ``direct effect'', is defined as an unconditional contrast between infection outcomes among treated and untreated individuals \citep{halloran1991study,halloran1997study,halloran2010design,halloran2016dependent}.  In an influential paper, \citet{hudgens2008toward} proposed a randomization design and a definition of the ``direct effect'' under interference in a clustered study population, along with effect estimators. Informally, the direct effect is defined as a contrast between the rate of infection for an individual under treatment versus no treatment, averaged over the conditional distribution of treatments to others in the same cluster \citep{vanderweele2011effect,savje2017average}.  
The direct effect estimand introduced by \citet{hudgens2008toward} has been applied in empirical analyses of randomized trials \citep[e.g.][]{perez2014assessing,buchanan2018assessing}.


The susceptiblity effect and direct effect are not the same.  However, they may appear to measure similar causal features of the effect of an intervention on individuals who receive it, especially under randomization. Informal descriptions of the direct effect imply comparability between treated and untreated individuals: \citet[][page 332]{halloran1991study} write, ``The direct effect of an intervention received by an individual is the difference betweeen the outcome in the individual with the intervention and what the outcome would have been without the intervention, all other things being equal''.  In the textbook \emph{Design and Analysis of Vaccine Studies}, \citet[][page 272]{halloran2010design} state ``An example of a direct effect is the reduction in the probability of becoming infected that results from being vaccinated, given exposure to infection.''
Writing of a randomized study design in which the direct effect is defined as the comparison of infection outcomes in treated individuals with untreated individuals, \citet[][page 334]{halloran1991study} state: ``After intervention, design I is the only design with comparable exposure to infection in the comparison groups''.  
Randomization ensures that on average, treated and untreated individuals do not vary systematically in their baseline characteristics.  
Indeed, \citet[][page 146]{halloran1995causal} write ``Under a random assignment of the vaccine to the population, then if everyone were exposed to infection, the average causal direct effect of the vaccine on the transmission probability would be estimated as the difference in the average outcomes in the unvaccinated and vaccinated individuals under the actual treatment assignment''.  In other words, when exposure is present, randomization ensures that the direct effect estimates the susceptibility effect.  

But even when treatment is randomized, exposure to infection can be systematically different among treated and untreated individuals during the study. Researchers have warned that this differential exposure can confound estimates of the ``direct effect'' of the intervention \citep{halloran1991direct,halloran1991study,struchiner1994malaria,halloran1995causal,halloran2010design,kenah2015semi,morozova2018risk}, but the relationship between the randomization design and the disease transmission process remains obscure \citep{struchiner2007randomization,van2013estimation,ohagan2014estimating}.  Do contrasts of infection outcomes between treated and untreated subjects, as proposed by \citet{hudgens2008toward} as the ``direct effect'', recover the susceptibility effect of the intervention when the population is clustered, treatment is randomized, and outcomes are contagious?



The purpose of this paper is to examine the meaning of the ``direct effect'' defined by \citet{hudgens2008toward} when infectious disease outcomes are transmissible in a study of potentially interacting individuals within clusters.  We first provide a formal definition of the causal susceptibility effect \citep{halloran1997study}, which is of primary interest in trials of infectious disease interventions.  We then briefly review the direct effect, and define three common randomization designs -- Bernoulli, block, and cluster randomization -- that may be employed in empirical trials of infectious disease interventions.  To compare the susceptibility and direct effects in a trial of an infectious disease intervention, we evaluate infection outcomes under a general structural model of infectious disease transmission in clusters that accommodates individually varying susceptibility to infection, infectiousness, and exogenous source of infection.  This type of structural model has found wide application in studies of infectious disease outcomes in clusters of individuals \citep{rhodes1996counting,longini1999markov,auranen2000transmission,oneill2000analyses,becker2003estimating,becker2004estimating,cauchemez2004bayesian,cauchemez2006investigating,cauchemez2009household,becker2006estimating,yang2006design,kenah2013non,kenah2015semi,tsang2015influenza,tsang2016individual,morozova2018risk}.  
We show that under some forms of randomization, the direct effect may not recover the sign of the true susceptibility effect of the intervention on the individual who receives it. In particular, when the intervention both helps protect treated individuals from infection, and helps prevent infected treated individuals from transmitting the infection to others, the direct effect can nevertheless be positive (indicating harm) under the randomization design proposed by \citet{hudgens2008toward}.  The results are derived using a probabilistic coupling argument that reveals stochastic dominance relations between infection outcomes under different treatment allocations.  These results substantially sharpen the claims of \citet{halloran1991direct} and \citet{struchiner2007randomization}, and generalize bias results for clusters of size two \citep{halloran2012causal,morozova2018risk}.


\section{Setting}

Consider a population of $N$ clusters, and let $n_i$ be the number of individuals in cluster $i$.  Suppose the outcome of interest is infection by an infectious disease that is transmissible between individuals within clusters, but not between clusters. Let $T_{ij}$ be the random infection time of subject $j$ and let $Y_{ij}(t)=\indicator{T_{ij}<t}$ be the indicator of prior infection. A subject $j$ is called susceptible at time $t$ if $Y_{ij}(t)=0$ and infected if $Y_{ij}(t)=1$.  The joint treatment vector $\x_i=(x_{i1},\ldots,x_{in_i})$ is allocated at baseline, $t=0$.  Following notation introduced by \citet{hudgens2008toward}, we will sometimes write the joint treatment allocation in cluster $i$ as $\x_i=(x_{ij},\x_{i(j)})$, where $x_{ij}$ is the treatment to subject $j$, and $\x_{i(j)}$ is the vector of treatment assignments to subjects other than $j$ in cluster $i$.


\subsection{Target parameter: susceptibility effect}

To define meaningful intervention effects for infectious disease outcomes, it is often necessary to consider a joint intervention on both the treatment assignment and exposure history of cluster members \citep{halloran1995causal,ohagan2014estimating}.  We use potential outcome notation \citep{rubin2005causal} to define causal effects.  
Let the infection status history of all subjects other than $j$ in cluster $i$ be denoted $\bH_{i(j)} = \left\{ Y_{ik}(s),\ k\neq j,\ s \ge 0 \right\}$, with a particular realization denoted by $\h_{i(j)}$. The infection history $\bH_{i(j)}$ is a vector of $n_i-1$ indicator functions denoting infection status for all times $0\le s < \infty$.  
Let $T_{ij}(\x_i,\h_{i(j)})$ be the potential infection time of $j$ when treatment is set to $\x_i$ and the infection history of individuals other than $j$ is set to $\h_{i(j)}$.  Let $Y_{ij}(t,\x_i,\h_{i(j)})=\indicator{T_{ij}(\x_i,\h_{i(j)})<t}$ be the corresponding potential infection outcome of subject $j$ at time $t$.  It is implicit that for fixed $t$, the potential infection outcome $Y_{ij}(t,\x_i,\h_{i(j)})$ does not depend on any element $k$ of $\h_{i(j)}$ when $T_{ik}>t$. In other words, infection of $k$ after $t$ does not affect infection of $j$ prior to $t$.  When $\x_i$ is a fixed treatment allocation and the infection history of other individuals $\bH_{i(j)}$ is allowed to arise naturally without intervention on infection history, we write $T_{ij}(\x_i) = T_{ij}(\x_i,\bH_{i(j)})$ and $Y_{ij}(t,\x_i) = \indicator{T_{ij}(\x_i,\bH_{i(j)}) < t}$. 
We regard potential infection outcomes as inherently stochastic: given a treatment allocation $\x_i$ and infection histories $\h_{i(j)}$, the potential infection time $T_{ij}(\x_i,\h_{i(j)})$ is a random variable. 



Following \citep{halloran1997study} and \citet{ohagan2014estimating}, we define the susceptibility effect as a contrast of the infection outcome of $j$ under treatment ($x_{ij}=1$) versus no treatment($x_{ij}=0$), while holding constant the treatments $\x_{i(j)}$ and infection histories $\h_{i(j)}$ of other cluster members.  Define the potential hazard of infection to subject $j$ in cluster $i$ at time $t$ as the instantaneous risk of infection at time $t$, given no infection up to $t$, holding other individuals' infection history $\h_{i(j)}$ and treatments $\x_{i(j)}$ constant:
\[ \lambda_{ij}\left(t,x_{ij},\x_{i(j)},\h_{i(j)}\right) = \lim_{\epsilon\to 0}\  \E\left[Y_{ij}\left(t+\epsilon,x_{ij},\x_{i(j)},\h_{i(j)}\right) \mid Y_{ij}\left(t,x_{ij},\x_{i(j)},\h_{i(j)}\right)=0\right] \]
when this limit exists.  The susceptibility hazard ratio (SHR) contrasts potential hazards under treatment versus no treatment of $j$, while holding all else -- including exposure to infection -- constant \citep{halloran1997study}:
\begin{equation}
  \SHR_{ij}\left(t,\x_{i(j)},\h_{i(j)}\right) = \frac{\lambda_{ij}\left(t,1,\x_{i(j)},\h_{i(j)}\right)}{\lambda_{ij}\left(t,0,\x_{i(j)},\h_{i(j)}\right)} .
  \label{eq:shr}
  \end{equation}
Informally, $\SHR_{ij}$ contrasts the instantaneous potential risk of infection of susceptible subject $j$ at time $t$ under treatment versus no treatment, while holding constant the treatments and infection histories of others.  
\citet[][Table 1]{halloran1997study} define the ``vaccine effect on susceptibility'' as $1-\SHR_{ij}$. 
Cluster and population-level susceptibility estimands may be defined as expectation of $\SHR_{ij}(t,\X_{i(j)},\bH_{i(j)})$, or as a ratio of expectations of the hazards.  
Analogous exposure-conditioned susceptibility effects can be defined on the risk difference and odds ratio scales \citep{ohagan2014estimating}.

\subsection{The ``direct effect'' in a randomized trial}

Define the expected individual infection outcome under join treatment $\x_i$ as $\overline{Y}_{ij}(t,\x_i) = \E[Y_{ij}(t,\x_i)]$, where expectation is with respect to the infection outcomes in cluster $i$. 
Let $\mathcal{X}^{n}=\{0,1\}^n$ be the set of all binary vectors of $n$ elements.  
We define causal estimands by comparing average infection outcomes under different treatment allocations to the cluster. These definitions are taken, with minor changes in notation, from \citet{hudgens2008toward}.  Define the \emph{individual average potential outcome} as
\begin{equation}
\overline{Y}_{ij}(t,x) = \sum_{\x_{i(j)}\in \mathcal{X}^{n_i-1}} \overline{Y}_{ij}(t,x,\x_{i(j)}) \Pr(\X_{i(j)}=\x_{i(j)} | X_{ij}=x) .
\label{eq:iapo}
\end{equation}
Informally, $\overline{Y}_{ij}(t,x)$ is the individual infection outcome under $x_{ij}=x$, averaged over the conditional distribution of treatments to the other individuals in cluster $i$.  
Define the \emph{cluster average potential outcome} as
$\overline{Y}_i(t,x) = n_i^{-1} \sum_{j=1}^{n_i} \overline{Y}_{ij}(t,x)$,
and the \emph{population average potential outcome} as  
$\overline{Y}(t,x) = N^{-1} \sum_{i=1}^{N} \overline{Y}_{i}(t,x)$.
\citet{hudgens2008toward} propose contrasts of these potential outcomes as causal estimands, which we rewrite in slightly different form.  Define the \emph{individual average direct effect} as 
$DE_{ij}(t) = \overline{Y}_{ij}(t,1) - \overline{Y}_{ij}(t,0)$, 
the \emph{cluster average direct effect} as
$DE_i(t) = n_i^{-1} \sum_{j=1}^{n_i} DE_{ij}(t)$,
and the \emph{population average direct effect} as 
$DE(t) = N^{-1} \sum_{i=1}^{N} DE_i(t)$.

\subsection{Randomization designs for clustered subjects}

A randomization design is a probability distribution that assigns the joint binary treatment vector $\x_i=(x_{i1}$, $\ldots$, $x_{in_i})$ within and across clusters.  

\begin{defn}[Bernoulli randomization]
The treatment is Bernoulli randomized if for every cluster $i$, 
the joint allocation 
$\x_i=(x_{i1},\ldots,x_{in_i})$ has probability $\Pr(\X_i=\x_i) = \prod_{j=1}^{n_i} p^{x_{ij}} (1-p)^{1-x_{ij}}$ for some probability $p$. 
\end{defn}


\begin{defn}[Block randomization]
The treatment is block-randomized if for every cluster $i$, 
the joint allocation 
$\x_i=(x_{i1},\ldots,x_{in_i})$ has probability 
  $\Pr(\X_i=\x_i) = \binom{n_i}{m_i}^{-1} \indicator{\sum_{j=1}^{n_i} x_{ij} = m_i}$ 
  where $0<m_i = \lfloor p n_i \rfloor$ for some probability $p>1/\min_i n_i$. 
\end{defn}


\begin{defn}[Cluster randomization]
The treatment is cluster randomized if for each cluster $i$, either all 
members of the cluster are treated, or all are untreated with probability 
$0 < p < 1$. That is, $\Pr(\X_i = (1,\ldots,1)) = p$ and  $\Pr(\X_i = (0,\ldots,0)) = 1-p$
for each cluster $i$ independently. 
\end{defn}

Block and cluster randomization designs induce dependencies in the treatment status of subjects in the same cluster. This means that the conditional treatment probability $\Pr(\X_{i(j)}=\x_{i(j)}|X_{ij}=x)$ in \eqref{eq:iapo} may differ for $x=1$ and $x=0$, and so the individual average risk difference $DE_{ij}(t)$ may not be an average of individualistic effects.  \citet{vanderweele2011effect} point out that the risk difference $DE(t)$ may suffer from difficulties in interpretation under block randomization, because it compares the outcome of a treated individual whose cluster contains $m_i-1$ others treated with an untreated individual whose cluster contains $m_i$ others treated.  \citet{savje2017average} call $DE(t)$ the ``average distribution shift effect'' because it ``captures the compound effect of changing a unit's treatment and simultaneously changing the experimental design''.  However, it remains unclear whether the direct effect $DE_{ij}(t)$ has a meaningful interpretation when interference arises via contagion.



\section{Approach}

Do the ``direct effect'' quantities $DE_{ij}(t)$, $DE_i(t)$, and $DE(t)$ above recover useful features of the susceptibility effect of interest in a randomized trial?  For example, if the treatment $x$ is a vaccine that truly helps prevent infection in the person who receives it when exposure to infection is held constant ($\SHR<1$), investigators conducting a randomized trial might want to know whether they should expect $DE(t)<0$.  To answer this question, we must specify more precisely the way that infection outcomes arise under contagion.  Epidemiologists have proposed structural models of infectious disease outcomes that formalize common ideas about the mechanism, or dynamics, of transmission in groups \citep{becker1989analysis,anderson1992infectious,andersson2000stochastic}. Many structural transmission models represent the individual risk (or hazard) of infection as an explicit function of individual treatments and possibly other covariates \citep{rhodes1996counting,longini1999markov,auranen2000transmission,oneill2000analyses,becker2003estimating,becker2004estimating,cauchemez2004bayesian,cauchemez2006investigating,becker2006estimating,yang2006design,kenah2013non,kenah2015semi,morozova2018risk}.  Structural models can be useful in both observational and randomized trials because they posit an explicit regression-style relationship linking covariates and infection outcome.  

We present a general structural model of infectious disease transmission based on the canonical stochastic susceptible-infective epidemic process \citep{becker1989analysis,andersson2000stochastic,diekmann2012mathematical}. This model, based on constructions by \citet{rhodes1996counting} and \citet{kenah2015semi}, captures the essential features of infectious disease transmission, and the effect of treatment on susceptibility and infectiousness.  In particular, this model represents the instantaneous risk (hazard) of infection experienced by subject $j$ in cluster $i$ as a non-decreasing step function whose jumps correspond to infections of other cluster members.  Conveniently, the susceptibility effect $\SHR$ corresponds explicitly to a parameter in this model. Recall that $\h_{i(j)}$ consists of the infection histories of individuals other than $j$: or $\h_{i(j)} = \{y_{ik}(t),\ k\neq j,\ t \ge 0\}$.  Let the hazard of infection experienced by a susceptible individual $j$ in cluster $i$ 
at time $t$ be 
\begin{equation}
  \lambda_{ij}(t,\x_i,\h_{i(j)}) = e^{x_{ij} \beta + \eta_{ij}} 
    \left(\alpha + \sum_{k=1}^{n_i} y_{ik}(t) e^{x_{ik} \gamma + \xi_{ik}} \right) 
\label{eq:haz}
\end{equation}
where $\beta$ is the effect of individual treatment $x_{ij}$, 
$\eta_{ij}$ is an individualistic susceptibility coefficient, 
$\alpha$ is the force of infection from outside the cluster, 
$\gamma$ is the infectiousness effect of the treatment $x_{ik}$ assigned to $k$ 
and $\xi_{ik}$ is an individualistic infectiousness coefficient for $k$.  The sum over $k$ in \eqref{eq:haz} does not include $k=j$ because $j$ cannot infect themselves.  Under this structural model, the susceptibility effect of interest \eqref{eq:shr} has a simple time-invariant form: $\SHR_{ij}(t,\x_{i(j)},\h_{i(j)}) = 
e^\beta$.

The structural transmission model \eqref{eq:haz} formalizes intuition about how interference arises for infectious disease outcomes.  The hazard of infection experienced by subject $j$ at time $t$ is a function of subject $k$'s features ($x_{ik}$ and $\xi_{ik}$) only when $k$ is currently infected ($y_{ik}(t)=1$).  As in \citet{hudgens2008toward}, the structural transmission model \eqref{eq:haz} obeys ``partial interference'' \citep{sobel2006randomized,halloran1991study,halloran1995causal}: the infection outcome for subject $j$ in cluster $i$ may depend on treatments and infection outcomes of other individuals in cluster $i$, but does not depend on subjects in clusters other than $i$.  Variations on this infection hazard model \eqref{eq:haz} have been used to model sources of disease transmission and for estimation of covariate effects on infection risk \citep{rhodes1996counting,auranen2000transmission,cauchemez2004bayesian,cauchemez2006investigating,kenah2013non,kenah2015semi,tsang2018transmissibility}, and as a conceptual model to evaluate the properties of risk ratios under contagion \citep{morozova2018risk}.  Figure \ref{fig:illustration} shows a schematic illustration of the transmission hazard model \eqref{eq:haz} for a cluster $i$ of size $n_i=4$ in which two subjects are treated.




\begin{figure}[t]
  \centering
  \includegraphics[width=\textwidth]{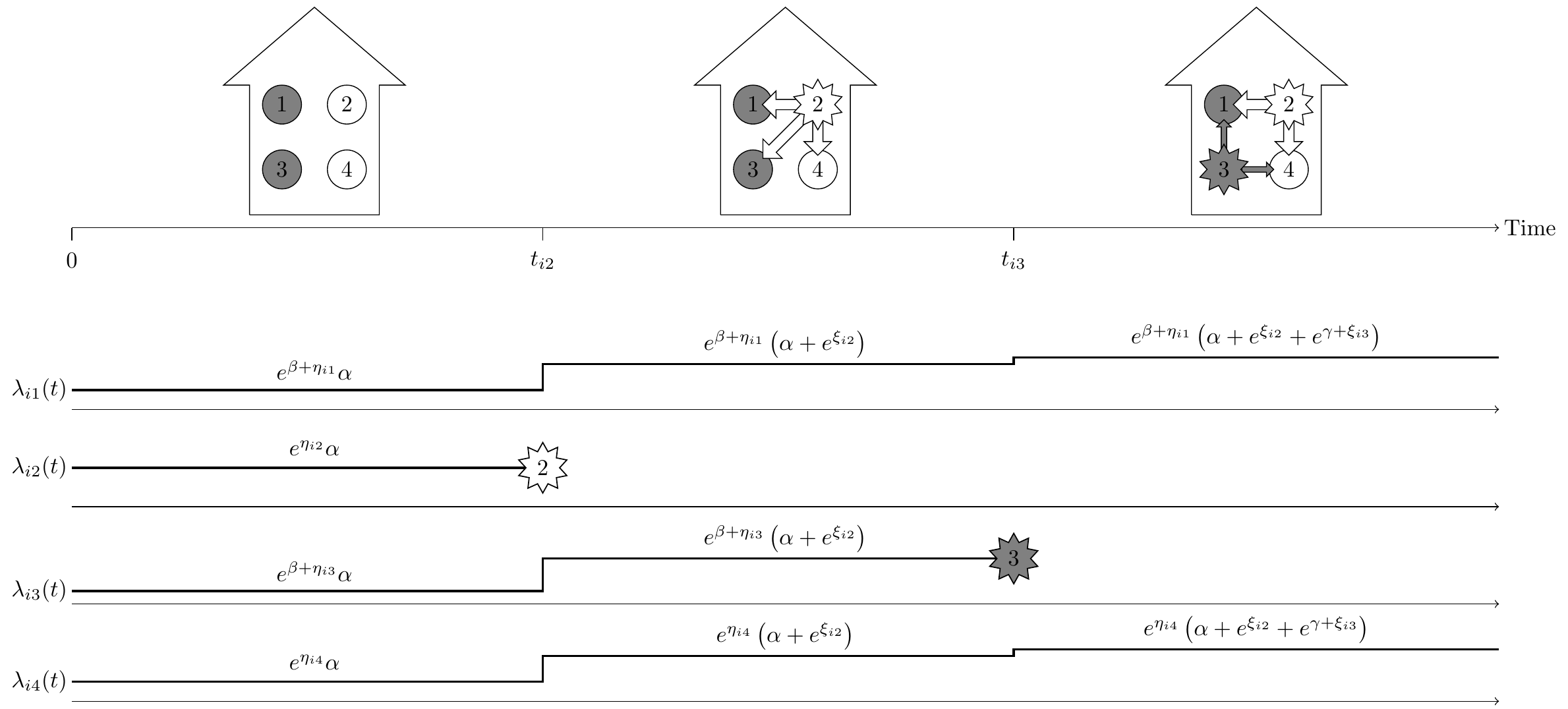}
  \caption{Illustration of the infectious disease transmission process and hazards \eqref{eq:haz} in a cluster $i$ of size $n_i=4$, where individuals 1 and 3 are treated (gray shading), while 2 and 4 are not.  Circles indicate susceptible individuals, and stars indicate infected individuals; arrows represent risk of transmission.  At time $t_{i2}$, subject 2 becomes infected, and thereafter transmits infection risk to 1, 3, and 4. Subsequently 3 becomes infected, and 1 and 4 are exposed to infection risk from both 2 and 3. The magnitude of this infection risk is related to the treatment status of the susceptible and infectious individuals.  At bottom, hazards of infection $\lambda_{ij}(t)$ are shown over time for each subject. The ``susceptibility'' effect of treatment is $\beta$, and the ``infectiousness'' effect is $\gamma$. 
  }
  \label{fig:illustration}
\end{figure}


\section{Results}

\subsection{DE under the null hypothesis of no susceptibility effect} 

If the direct effect is to serve as a useful estimand for researchers interested in learning about the causal effect of the intervention on the subject who receives it, we should expect that $DE_{ij}(t)=0$ when $\beta=0$, since the treatment has no effect on the infection risk of an individual who receives it. 
We begin by studying the properties of the average individual direct effect 
$DE_{ij}(\mathcal{T})$ under the three randomization designs.  We assume that the exogenous (community) force of infection $\alpha$ is positive, and 
$\mathcal{T}>0$ is a follow-up time at which infection outcomes are measured, so that at least one infection in each cluster arises with positive probability. 

Bernoulli randomization gives concordance between $\beta=0$ and the direct effect. 

\begin{prop}[DE under Bernoulli randomization]
Suppose $\beta=0$ and treatment assignment is Bernoulli randomized. Then $DE_{ij}(\mathcal{T})=0$.  \label{prop:bernoulli}
\end{prop}

\noindent In contrast, the direct effect has the opposite sign as the infectiousness effect $\gamma$ when $\beta=0$ under block randomization. 

\begin{prop}[DE under block randomization]
Suppose $\beta=0$ and treatment assignment is block-randomized. 
If $\gamma<0$ then $DE_{ij}(\mathcal{T})>0$; if $\gamma=0$ then $DE_{ij}(\mathcal{T})=0$; 
and if $\gamma>0$ then $DE_{ij}(\mathcal{T})<0$. 
\label{prop:block}
\end{prop}

\noindent The direct effect has the same sign as $\gamma$ when $\beta=0$ under cluster randomization.  

\begin{prop}[DE under cluster randomization]
Suppose $\beta=0$ and treatment assignment is cluster randomized. 
If $\gamma<0$ then $DE_{ij}(\mathcal{T})<0$; if $\gamma=0$ then $DE_{ij}(\mathcal{T})=0$; 
and if $\gamma>0$ then $DE_{ij}(\mathcal{T})>0$. 
\label{prop:cluster}
\end{prop}

Propositions \ref{prop:bernoulli}-\ref{prop:cluster} compare averaged expectations of infection outcomes for subject $j$ in cluster $i$. However, computing the expectation $\overline{Y}_{ij}(t,x,\x_{i(j)})$ for particular values of $x$ and $\x_{i(j)}$ is intractable, so an explicit comparison of average individual potential infection outcomes under different treatment allocations cannot be made analytically. Instead, we will use tools from the theory of probabilistic coupling \citep{den2012probability,ross1996stochastic} to exhibit stochastic dominance relations between infection outcomes under different treatment allocations to facilitate the comparison.  

\begin{defn}[Coupling]
A coupling of two random variables $Y^0$ and $Y^1$ both taking values in 
$(\Omega,\mathcal{F})$ is any pair of random variables 
$(\tilde Y^0,\tilde Y^1)$ taking values in 
$(\Omega\times \Omega, \mathcal{F}\otimes\mathcal{F})$ whose marginal distributions are identical to those of $Y^0$ and $Y^1$ respectively, i.e. $Y^0 \overset{d}{=} \tilde Y^0$ 
and $Y^1 \overset{d}{=} \tilde Y^1$.  
\label{defn:coupling}
\end{defn}

Typically the variables $\tilde Y^0$ and $\tilde Y^1$ are dependent. To study the relationship of infection outcomes under different treatment scenarios, a notion of dominance will be necessary. 

\begin{defn}[Stochastic dominance]
The real-valued random variable $Y^1$ stochastically dominates $Y^0$ if 
$\Pr(Y^1<y) \le \Pr(Y^0<y)$ for all $y\in\R$. 
\end{defn}

If $Y^1$ stochastically dominates $Y^0$ and vice versa, the variables are equal in distribution.  If $Y^1$ stochastically dominates $Y^0$, then $\E[Y^1] \ge \E[Y^0]$.  The following Lemma, proved by e.g. \citet[][pages 409--410]{ross1996stochastic}, provides a framework for establishing stochastic dominance through the construction of a coupling. 

\begin{lem}[Coupling and stochastic dominance]
The real-valued random variable $Y^1$ stochastically dominates $Y^0$ if and 
only if there is a coupling $(\tilde Y^0$, $\tilde Y^1)$ of $Y^0$ and $Y^1$ 
such that $\Pr(\tilde Y^1 \ge \tilde Y^0) = 1$.
\label{lem:dominance}
\end{lem}

To begin proving Propositions \ref{prop:bernoulli}-\ref{prop:cluster}, define the vectors of stochastic potential outcomes of all subjects under two different joint treatments allocations $\x_i^1$ and $\x_i^0$ as  $\Y_i(t,\x_i^1)=\left(Y_{i1}(t,\x_i^1), \ldots, Y_{in_i}(t,\x_i^1)\right)$ and $\Y_i(t,\x_i^0)=\left(Y_{i1}(t,\x_i^0), \ldots, Y_{in_i}(t,\x_i^0)\right)$.  Corresponding to these potential outcomes, we will construct two coupled outcome processes with $\beta=0$, denoted $\tilde \Y_i^1(t)$ and $\tilde \Y_i^0(t)$, under treatment vectors $\x_i^1$ and $\x_i^0$ respectively.  The order of infections in both processes is the same, but the times of infection may be different.

Let $S_l$ and $I_l$ be the set of subjects that are susceptible and infectious, respectively, just before the $l$th infection.  Let $\tilde{W}_l^1$ and $\tilde{W}_l^0$ be the waiting times to the next infection in the coupled processes under $\x_i^1$ and $\x_i^0$ respectively.  
    Define the waiting time cumulative distribution functions 
   \[ F_l(w) = 1-\exp\left[ -w \sum_{a\in S_l} e^{\eta_{ia}} \left( \alpha + \sum_{b\in I_l} e^{\gamma x_{ib}^1 + \xi_{ib}}\right)\right] \]
  and 
  \[ G_l(w) =  1-\exp\left[ -w \sum_{a\in S_l} e^{\eta_{ia}} \left( \alpha + \sum_{b\in I_l} e^{\gamma x_{ib}^0 + \xi_{ib}}\right)\right], \]
  where sums over empty sets are interpreted as zero.  Let $\tilde T_{il}^1$ and $\tilde T_{il}^0$ be the time of infection of subject $l$ under treatments $\x_i^1$ and $\x_i^0$ respectively, with $\tilde T_{i0}^1=\tilde T_{i0}^0=0$. Likewise define the corresponding infection indicators $\tilde Y_{il}^1(t) = \indicator{\tilde T_{il}^1<t}$ and $\tilde Y_{il}^0(t) = \indicator{\tilde T_{il}^0<t}$.  The following algorithm constructs the joint outcome functions 
$\tilde\Y_i^1(t)=(\tilde Y_{i1}^1(t), \ldots, \tilde Y_{in_i}^1(t))$ and  
$\tilde\Y_i^0(t)=(\tilde Y_{i1}^0(t), \ldots, \tilde Y_{in_i}^0(t))$
under treatment vectors $\x_i^1$ and $\x_i^0$ respectively.  We show below that $(\tilde \Y_i^1(t),\tilde \Y_i^0(t))$ is a coupling of the potential infection outcomes $\Y_i(t,\x^1)$ and $\Y_i(t,\x^0)$.

\begin{algorithm}[H]
  \begin{algorithmic}[0] 
    \State $S_1 \gets \{1, \ldots, n_i \}$  \Comment{Initialize susceptibles}
    \State $I_1 \leftarrow \emptyset$ \Comment{Initialize infectives} 
    \For{$l\gets 1,\ldots,n_i$}
      \State $U_l \sim \text{Uniform}(0,1)$ 
      \State $\tilde W_l^1 \gets F_l^{-1}(U_l)$ \Comment{Set waiting times to next infection} 
      \State $\tilde W_l^0 \gets G_l^{-1}(U_l)$
      \For{$v \in S_l$}  \Comment{Set subject infection probabilities}
         \State $p_v \gets e^{\eta_{iv}}/\sum_{a\in S_l} e^{\eta_{ia}}$
      \EndFor
      \State $V_l \sim \text{Multinomial}(S_l, \{p_v:\ v\in S_l\})$ \Comment{Choose next infected subject} 
      \State $\tilde T_{iV_l}^1 \gets \tilde T_{iV_{l-1}}^1 + \tilde W_l^1$ \Comment{Set infection time}
      \State $\tilde T_{iV_l}^0 \gets \tilde T_{iV_{l-1}}^0 + \tilde W_l^0$
      \State $\tilde Y_{iV_l}^1(t) \gets \indicator{\tilde T_{iV_l}^1<t}$ \Comment{Define infection indicator function}
      \State $\tilde Y_{iV_l}^0(t) \gets \indicator{\tilde T_{iV_l}^0<t}$
      \State $S_l \gets S_{l-1} \setminus \{ V_l \}$ \Comment{Update susceptibles}
      \State $I_l \gets I_{l-1} \cup \{ V_l \}$ \Comment{Update infectives}
    \EndFor

\end{algorithmic}
\caption{Construction of the coupling for cluster $i$.}
\label{alg:coupling}
\end{algorithm}

Algorithm \ref{alg:coupling} generates two sets of infection outcomes, one corresponding to the joint treatment $\x_i^1$ and one to the joint treatment $\x_i^0$, by constructing waiting times to infection of each subject, and which subject is infected at each step.  The key insight is that under the infection hazard model \eqref{eq:haz}, the waiting times $\tilde{W}_{il}^1$ and $\tilde{W}_{il}^0$ depend on treatments of already-infected individuals, but because $\beta=0$, selection of the next infected individual does not depend on treatments of yet-uninfected subjects.  This fact permits construction of two dependent infection processes whose timing differs, but where the order of infections is identical.

\begin{lem}[Construction of the coupling]
  When $\beta=0$, the variables $(\tilde\Y_i^1(t),\tilde\Y_i^0(t))$ constructed by Algorithm \ref{alg:coupling} constitute a coupling of the potential infection outcomes $\Y_i(t,\x^1)$ and $\Y_i(t,\x^0)$.
\label{lem:coupling}
\end{lem}

\begin{proof}[Proof of Lemma \ref{lem:coupling}] 
 We will show that $(\tilde\Y_i^1(t),\tilde\Y_i^0(t))$ is a coupling of $\Y_i(t,\x^1)$ and $\Y_i(t,\x^0)$ satisfying Definition \ref{defn:coupling}. 
First, the waiting time distribution functions $F_l(w)$ and $G_l(w)$ are 
monotonically increasing in $w$, so the random waiting time 
$\tilde W_l^1 = F_l^{-1}(U_l)$ has distribution function $F_l(w)$ and 
$\tilde W_l^0 = G_l^{-1}(U_l)$ has distribution function $G_l(w)$ 
\citep{devroye1986random}.  
Because the same uniform variable $U_l$ is used to generate both waiting times $\tilde{W}_l^1$ and $\tilde{W}_l^0$, these variables, and hence the infection times $\tilde T_{ij}^1$ and $\tilde T_{ij}^0$, and outcomes $\tilde Y_{ij}^1(t)$ and $\tilde Y_{ij}^0(t)$, are dependent. 
The joint mass function of the $l$th infected subject $V_l$ and the cumulative 
distribution function of the waiting time $\tilde W_l^1$ to this infection is, 
by construction,  
\begin{equation}
\begin{split} 
  \Pr(V_l=v,\tilde W_l^1<w) 
  &= \frac{e^{\eta_{iv}} }{\sum_{a\in S_l} e^{\eta_{ia}}} 
     \left[1-\exp\left[
       - w \sum_{a\in S_l} e^{\eta_{ia}} 
       \left( \alpha + \sum_{b\in I_l} e^{\gamma x_{ib}^1 + \xi_{ib}}\right) 
       \right]\right].
\end{split}
\label{eq:jointmasscdf}
\end{equation}
Differentiating \eqref{eq:jointmasscdf} with respect to $w$, we find that the 
joint likelihood of the newly infected subject $V_l=v$ and the waiting time $w$ to the $l$th 
infection is
\begin{equation}
\begin{split}
  \tilde L_{il}^1(v,w) &= e^{\eta_{iv}}\left( \alpha + \sum_{b\in I_l} e^{\gamma x_{ib}^1 + \xi_{ib}} \right) 
  \exp\left[- w \sum_{a\in S_l} e^{\eta_{ia}} 
    \left( \alpha + \sum_{b\in I_l} e^{\gamma x_{ib}^1 + \xi_{ib}}\right) 
      \right] \\
      &= e^{\eta_{iv}}\left( \alpha 
      + \sum_{b=1}^n \tilde y_{ib}^1(\tilde t_{iv}^1) e^{\gamma x_{ib}^1 + \xi_{ib}} \right)
      \exp\left[-w\sum_{a=1}^n (1-\tilde y_{ia}^1(\tilde t_{iv}^1)) e^{\eta_{ia}} 
        \left(\alpha + \sum_{b=1}^n \tilde y_{ib}^1(\tilde t_{iv}^1) e^{\gamma x_{ib}^1 + \xi_{ib}}\right)
      \right] \\
      &= \tilde \lambda_{iv}^1(\tilde t_{iv}^1) \exp\left[
  - w \sum_{a=1}^n (1-\tilde y_{ia}^1(\tilde t_{iv}^1)) \tilde \lambda_{ia}^1(\tilde t_{iv}^1) 
    \right]
\end{split}
\end{equation}
where $\tilde \lambda_{ij}^1(t)$ is \eqref{eq:haz} with $\tilde \y_{i}^1(t)$ and $\x_i^1$ replacing $\y_{i}(t)$ and $\x_i$ respectively.  
Let $\tilde L_{i}^1(\tilde\y_i^1)$ be the likelihood of the full realization of
$\tilde \y_i^1(t) = (\tilde y_{i1}^1(t),\ldots,\tilde y_{in_i}^1(t))$
with 
$\tilde \T^1 = (\tilde t_{i1}^1, \ldots, \tilde t_{in_i}^1)$, 
$\tilde \W^1 = (\tilde w_{i1}^1, \ldots, \tilde w_{in_i}^1)$, and 
$\tilde \V = (\tilde v_{i1}^1, \ldots, \tilde v_{in_i}^1)$. 
Recall that by construction, $\tilde w_{ik}^1 = \tilde t_{iv_k}^1 - \tilde t_{v_{i,k-1}}^1$.  
The likelihood of the constructed process is 
\begin{equation}
\begin{split}
  \tilde L(\tilde \y^1(t)) 
  &= \prod_{k=1}^{n_i} \tilde \lambda_{iv_k}^1(\tilde t_{iv_k}^1) 
  \exp\left[- w_k\sum_{j=1}^{n_i} (1-\tilde y_{ij}^1(\tilde t_{iv_k}^1))\tilde \lambda_{ij}^1(\tilde t_{iv_k}^1)\right] \\
  &= \left( \prod_{j=1}^{n_i} \tilde \lambda_{ij}^1(\tilde t_{ij}^1) \right)
    \exp\left[- \sum_{j=1}^{n_i} \sum_{k=1}^{n_i} 
    \int_{\tilde t_{iv_{k-1}}^1}^{\tilde t_{iv_k}^1} (1-\tilde y_{ij}^1(t)) \tilde \lambda_{ij}^1(t) \dx{t} \right] \\ 
    &= \left( \prod_{j=1}^{n_i} \tilde \lambda_{ij}^1(\tilde t_{ij}^1) \right)
    \exp\left[ -\sum_{j=1}^{n_i} \int_{0}^{\tilde t_{ij}^1} \tilde \lambda_{ij}^1(t) \dx{t} \right] \\
  &= L(\tilde \y^1(t))
\end{split}
\end{equation}
where $L(\tilde \y_i^1(t))$ is the likelihood 
of the original process. 
Therefore the constructed outcome vector $\tilde \Y_i(t,\x_i^1)$ is equal in distribution to the potential outcome vector $\Y_i(t,\x_i^1)$, and it follows that 
$\tilde Y_{ij}(t,\x_i^1)$ is equal in distribution to $Y_{ij}(t,\x_i^1)$.  By the same reasoning, 
$\tilde Y_{ij}(t,\x_i^0)$ is equal in distribution to $Y_{ij}(t,\x_i^0)$.  Therefore by Definition \ref{defn:coupling}, 
$(\tilde Y_{ij}(t,\x_i^1),\tilde Y_{ij}(t,\x_i^0))$ is a coupling of $Y_{ij}(t,\x_i^1)$ and 
$Y_{ij}(t,\x_i^0)$.  
\end{proof}

With this coupling, we can deduce stochastic order relations in infection outcomes under particular different joint treatments $\x_i^1$ and $\x_i^0$ when $\beta=0$.  The proofs of Propositions \ref{prop:bernoulli}-\ref{prop:cluster} exhibit these order relations, under the three randomization designs.  In each case we focus, without loss of generality, on a particular subject $j$. 

\begin{proof}[Proof of Proposition~\ref{prop:bernoulli}] 
  First, let $\x_i^1=(x_{ij}=1,\x_{i(j)})$ and $\x_i^0=(x_{ij}=0,\x_{i(j)})$ be joint treatment allocations that are identical except that for subject $j$, $x_{ij}^1=1$ and $x_{ij}^0=0$.  Then by Lemma \ref{lem:coupling}, $(\tilde\Y_i^1(t),\tilde\Y_i^0(t))$ is a coupling of $\Y_i(t,\x^1)$ and $\Y_i(t,\x^0)$ under $\beta=0$. Whenever $j$ is uninfected, $F_l(w) = G_l(w)$, so $\tilde{T}_{il}^1 = \tilde{T}_{il}^0$ and so $\tilde{Y}_{ij}^1(t) = \tilde{Y}_{ij}^0(t)$.  Therefore by Lemma \ref{lem:dominance} $Y_{ij}(t,\x_i^1)$ stochastically dominates $Y_{ij}(t,\x_i^0)$ and vice versa.  It follows that $Y_{ij}(t,\x_i^1)$ and $Y_{ij}(t,\x_i^0)$ are equal in distribution, so $\overline{Y}_{ij}(\mathcal{T},1,\x_{i(j)}) = \overline{Y}_{ij}(\mathcal{T},0,\x_{i(j)})$. Now consider a Bernoulli randomized treatment allocation $\X_{i(j)}$ to subjects other than $j$ in cluster $i$.  Under Bernoulli randomization, the distribution of $\X_{i(j)}$ is invariant to conditioning on $X_{ij}=x_{ij}$.  By the definition of the individual average potential infection outcome \eqref{eq:iapo}, 
\begin{equation*}
  \begin{split}
    \overline{Y}_{ij}(\mathcal{T},1) &= \sum_{\x_{i(j)}\in \mathcal{X}^{n_i-1}} \overline{Y}_{ij}(\mathcal{T},1,\x_{i(j)}) \prod_{k\neq j} p^{x_{ik}}(1-p)^{1-x_{ik}} \\
    &= \sum_{\x_{i(j)}\in \mathcal{X}^{n_i-1}} \overline{Y}_{ij}(\mathcal{T},0,\x_{i(j)}) \prod_{k\neq j} p^{x_{ik}}(1-p)^{1-x_{ik}} \\
    &= \overline{Y}_{ij}(\mathcal{T},0), 
  \end{split}
\end{equation*}
and so $DE_{ij}(\mathcal{T})=0$ as claimed.
\end{proof}

\noindent The proof of Proposition \ref{prop:block} proceeds similarly, but we evaluate differences in potential infection outcomes of subject $j$ when $j$ and $k\neq j$ have opposite treatments, with other subjects' treatments held constant.  

\begin{proof}[Proof of Proposition~\ref{prop:block}]
Let $\mathcal{X}^n_m$ be the set of all binary $n$-vectors with $m$ positive elements. 
  First, we deduce a stochastic order relation for a particular treatment allocation in which $j$ and $k$ have opposite treatments. 
Let $\z\in \mathcal{X}^{n_i-2}_{m_i-1}$ and for $j\neq k$ define $\x_i^1=(x_{ij}=1,x_{ik}=0,\x_{i(jk)}=\z)$ and $\x_i^0=(x_{ij}=0,x_{ik}=1,\x_{i(jk)}=\z)$. 
When $\gamma=0$, $F_l(w)=G_l(w)$ for all $l$ and all $w$. Therefore 
$\tilde T_{il}^1=\tilde T_{il}^0$ for all $l$ and so 
$\tilde Y_{ij}^1(t)=\tilde Y_{ij}^0(t)$ for all $t$.
Then $Y_{ij}(t,\x_i^1)$ is equal in distribution to $Y_{ij}(t,\x_i^0)$ for all $t$  
and so $\overline{Y}_{ij}(t,\x_i^1)=\overline{Y}_{ij}(t,\x_i^0)$.  
When $\gamma<0$, note that $\tilde Y_{ij}^1(t) \ge \tilde Y_{ij}^0(t)$ for all $t$ 
if and only if $\tilde T_{ij}^1 \le \tilde T_{ij}^0$.  
Suppose without loss of generality that in the coupled processes, subjects are relabeled in order of their infection in the constructed process, so the $j$th infection occurs in subject $j$, $v_j=j$. Likewise the $k$th infection occurs in subject $k$, so $v_k=k$.  
Two cases are of interest. First, when $j<k$ we have $x_{il}^1=x_{il}^0$ for every $l\le j$, and so $F_l(w)=G_l(w)$ for $l\le j<k$ and all $w$. 
Therefore, 
\begin{equation}
  \tilde T_{ij}^1 = \sum_{l=1}^{j} \tilde W_l^1 
              = \sum_{l=1}^{j} F_l^{-1}(U_l)
              = \sum_{l=1}^{j} G_l^{-1}(U_l) 
              = \sum_{l=1}^{j} \tilde W_l^0 
              = \tilde T_{ij}^0.
\end{equation}
Second, when subject $k$ is infected first, or $k<j$, we have $F_l(w)=G_l(w)$ 
for $l<k$.  However, for subjects $r$ infected after $k$ ($r>k$), we have 
\begin{equation}
\begin{split}
  F_r(w) &= 1-\exp\left\{ -w \sum_{a\in S_r} e^{\eta_{ia}} \left( \alpha + \sum_{b\in I_r} e^{\gamma x_{ib}^1 + \xi_{ib}}\right) \right\} \\
  &= 1-\exp\left\{ -w \sum_{a\in S_r} e^{\eta_{ia}} \left( \alpha + \sum_{\substack{b\in I_r\\ b\neq k}} e^{\gamma x_{ib}^0 + \xi_{ib}} + e^{\xi_{ik}} \right) \right\} \\
  &> 1-\exp\left\{ -w \sum_{a\in S_r} e^{\eta_{ia}} \left( \alpha + \sum_{\substack{b\in I_r\\ b\neq k}} e^{\gamma x_{ib}^0 + \xi_{ib}} + e^{\gamma + \xi_{ik}} \right) \right\} \\
  &= 1-\exp\left\{ -w \sum_{a\in S_r} e^{\eta_{ia}} \left( \alpha + \sum_{b\in I_r} e^{\gamma x_{ib}^0 + \xi_{ib}} \right) \right\} \\
&= G_r(w)
\end{split}
\end{equation}
for all $w$. Therefore $F^{-1}_r(U_r) < G^{-1}_r(U_r)$ by monotonicity of 
$F_r(w)$ and $G_r(w)$, so the constructed infection times are 
\begin{equation}
\begin{split}
  \tilde T_{ij}^1 &= \sum_{l=1}^{j} \tilde W_l^1 
              = \sum_{l=1}^{j} F_l^{-1}(U_l) 
              = \sum_{l=1}^{k} G_l^{-1}(U_l) + \sum_{r=k+1}^{j} F_r^{-1}(U_r) \\
              &< \sum_{l=1}^{k} G_l^{-1}(U_l) + \sum_{r=k+1}^{j} G_r^{-1}(U_r) 
              = \sum_{l=1}^{j} \tilde W_l^0 
              = \tilde T_{ij}^0  .
  \end{split}
\end{equation}
Therefore 
$\tilde T_{ij}^1 \le \tilde T_{ij}^0$ and hence 
$\Pr(\tilde Y_{ij}^1(t) \ge \tilde Y_{ij}^0(t))=1$.  By Lemma \ref{lem:dominance}, 
$Y_{ij}(t,\x^1)$ stochastically dominates $Y_{ij}(t,\x^0)$ for all $t>0$.  
Because infection of subject $k$ before subject $j$ occurs with 
positive probability, it follows that the expected values of the
potential infection outcomes obey $\overline{Y}_{ij}(t,\x^1)>\overline{Y}_{ij}(t,\x^0)$.
The case $\gamma>0$ is the same as for $\gamma<0$, with inequalities switched. 

In summary, if $\gamma<0$ then $\overline{Y}_{ij}(t,\x_i^1)>\overline{Y}_{ij}(t,\x_i^0)$; 
  if $\gamma=0$ then $\overline{Y}_{ij}(t,\x_i^1)=\overline{Y}_{ij}(t,\x_i^0)$; 
  and if $\gamma>0$ then $\overline{Y}_{ij}(t,\x_i^1)<\overline{Y}_{ij}(t,\x_i^0)$.  
With these intermediate results in hand, we assess the role of the block randomization design.  

  Now let $\z$ be a binary vector of length $n_i-1$ with $m_i-1$ positive elements. Define 
  $\mathcal{P}_{i}(\z) = \{ \w\in\{0,1\}^{n_i-1}:\ (\w'\z,\w'\mathbf{1})=(m_i-1,m_i) \}$
 as the set of $n_i-m_i$ binary vectors $\w$ of length $n_i-1$ for which all positive elements of $\z$ are also positive in $\w$, and in addition $\w$ contains one more positive element.  
  Using this definition, and the combinatorial identity
\begin{equation}
  \binom{n_i-1}{m_i} = \frac{n_i-m_i}{m_i} \binom{n_i-1}{m_i-1}, 
  \label{eq:binom}
\end{equation}
  we can decompose a sum over allocations of $m_i$ treatments to $n_i-1$ subjects into a sum over allocations of $m_i-1$ treatments to $n_i-1$ subjects, and an additional allocation of treatment to one more,
\begin{equation}
  \sum_{\w\in\mathcal{X}^{n_i-1}_{m_i}} \overline{Y}_{ij}(t,0,\w) = \frac{1}{m_i} \sum_{\z\in\mathcal{X}^{n_i-1}_{m_i-1}} \sum_{\w\in\mathcal{P}_i(\z)} \overline{Y}_{ij}(t,0,\w) .
  \label{eq:partner}
\end{equation}
The factor $1/m_i$ appears in the right-hand side above because there are $m_i$ allocations $\z$ for which a given $\w\in \mathcal{P}_i(\z)$ is compatible; the double sum over-counts allocations by a factor of $m_i$.  Using this fact, we expand $DE_{ij}(\mathcal{T})$ into a sum over allocations to subjects other than $j$, 
  \begin{equation}
    \begin{split}
      DE_{ij}(\mathcal{T}) 
      &= \binom{n_i-1}{m_i-1}^{-1} \sum_{\z\in\mathcal{X}^{n_i-1}_{m_i-1}} \overline{Y}_{ij}(\mathcal{T},1,\z) - \binom{n_i-1}{m_i}^{-1} \sum_{\w\in\mathcal{X}^{n_i-1}_{m_i}} \overline{Y}_{ij}(\mathcal{T},0,\w)  \\
      &= \binom{n_i-1}{m_i-1}^{-1} \left( \sum_{\z\in\mathcal{X}^{n_i-1}_{m_i-1}} \overline{Y}_{ij}(\mathcal{T},1,\z) - \frac{m_i}{n_i-m_i} \sum_{\w\in\mathcal{X}^{n_i-1}_{m_i}} \overline{Y}_{ij}(\mathcal{T},0,\w)  \right) \\
      &= \binom{n_i-1}{m_i-1}^{-1} \frac{1}{n_i-m_i} \sum_{\z\in\mathcal{X}^{n_i-1}_{m_i-1}} \left( (n_i-m_i) \overline{Y}_{ij}(\mathcal{T},1,\z) - \sum_{\w\in\mathcal{P}_i(\z)} \overline{Y}_{ij}(\mathcal{T},0,\w)  \right) \\
      &= \binom{n_i-1}{m_i-1}^{-1} \frac{1}{n_i-m_i} \sum_{\z\in\mathcal{X}^{n_i-1}_{m_i-1}} \sum_{\w\in\mathcal{P}_i(\z)} \left( \overline{Y}_{ij}(\mathcal{T},1,\z) - \overline{Y}_{ij}(\mathcal{T},0,\w)  \right) \\
    \end{split}
    \label{eq:decompose}
  \end{equation}
  where the first equality follows from \eqref{eq:iapo} under block randomization with $m_i$ of $n_i$ subjects treated, the second by \eqref{eq:binom}, the third by \eqref{eq:partner}, and the fourth because there are $n_i-m_i$ terms in the sum over $\w\in\mathcal{P}(\z)$.  
  Therefore, $DE_{ij}(\mathcal{T})$ can be expressed as a sum of contrasts between the average outcome of $j$ under joint treatments $(1,\z)$ and $(0,\w)$ where $\w$ is the same as $\z$, but with one additional treated subject.  Each contrast in the last line of \eqref{eq:decompose} has sign as given above, and the result follows. 
\end{proof}

\noindent The proof of Proposition \ref{prop:cluster} is very similar and is presented in the Supplement.  
Three final results generalize the results for the individual average direct effect $DE_{ij}(\mathcal{T})$ to the cluster and population average direct effect estimands.  The proofs, which rely only on Propositions \ref{prop:bernoulli}, \ref{prop:block}, and \ref{prop:cluster} and the definitions of $DE_i(\mathcal{T})$ and $DE(\mathcal{T})$, are omitted.  

\begin{cor}[Cluster and population average $DE$ under Bernoulli randomization]
  Suppose $\beta=0$ and treatment assignment is Bernoulli randomized.  Then $DE_i(\mathcal{T})=DE(\mathcal{T})=0$. 
\label{cor:multiple:bernoulli}
\end{cor}

\begin{cor}[Cluster and population average $DE$ under block randomization]
Suppose $\beta=0$, treatment assignment is block randomized. 
If $\gamma<0$ then $DE_i(\mathcal{T})>0$ and $DE(\mathcal{T})>0$;
if $\gamma=0$ then $DE_i(\mathcal{T})=DE(\mathcal{T})=0$; and
if $\gamma>0$ then $DE_i(\mathcal{T})<0$ and $DE(\mathcal{T})<0$.
\label{cor:multiple:block}
\end{cor}

\begin{cor}[Cluster and population average $DE$ under cluster randomization]
Suppose $\beta=0$ and treatment assignment is cluster randomized. 
If $\gamma<0$ then $DE_i(\mathcal{T})<0$ and $DE(\mathcal{T})<0$;
if $\gamma=0$ then $DE_i(\mathcal{T})=DE(\mathcal{T})=0$; and
if $\gamma>0$ then $DE_i(\mathcal{T})>0$ and $DE(\mathcal{T})>0$.
\label{cor:multiple:cluster}
\end{cor}

\subsection{Simulation Study}

We investigate the properties of the population average direct effect as the true infectiousness effect $\gamma$ changes.  The hazard of infection takes the form of \eqref{eq:haz} where the null hypothesis is $\beta=0$ and we investigate $DE(\mathcal{T})$ as a function of $\gamma\in[-2,2]$.  The exogenous force of infection is $\alpha=0.01$, the individual susceptibility coefficients $\eta_{ij}$ are independent $\text{Normal}(\mu_\eta,\sigma^2_\eta)$ and infectiousness coefficients $\xi_{ij}$ are independent $\text{Normal}(\mu_\xi,\sigma^2_\xi)$.  
Unless otherwise noted, the cluster size $n_i$ is $2+\text{Poisson}(2)$, the observation time is $\mathcal{T}=10$, and all subjects were uninfected at baseline, $Y_{ij}(0)=0$. 
The Supplement provides additional details about the simulation setting. 

\begin{figure}[t]
  \centering
	\includegraphics[scale=0.7]{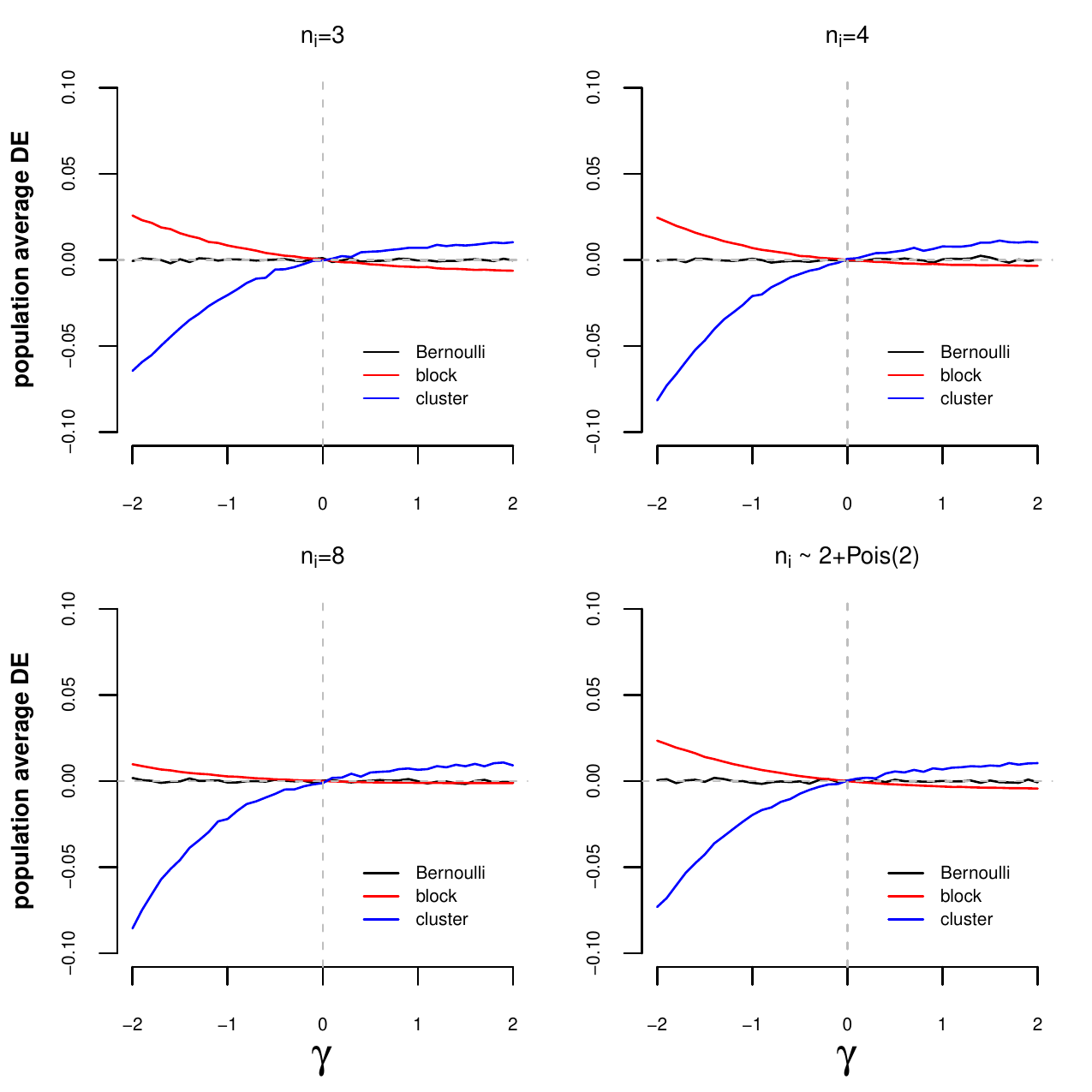} \\
  \caption{Simulation results for $DE(\mathcal{T})$ under the null hypothesis of no susceptibility effect $\beta=0$, as a function of the infectiousness effect $\gamma$, for different randomization designs and cluster size distributions.  Bernoulli randomization (black) recovers $DE(\mathcal{T})=0$ for any $\gamma$. Block randomization (red) shows $DE(\mathcal{T})$ has the opposite sign as $\gamma$, and cluster randomization (blue) shows $DE(\mathcal{T})$ has the same sign as $\gamma$. }
  \label{fig:sim1}
\end{figure} 

Figure \ref{fig:sim1} shows simulation results validating the analytic derivations above.  Under Bernoulli randomization $DE(\mathcal{T})$ is zero for any $\gamma$; under block randomization it has the opposite sign as $\gamma$; and under cluster randomization it has the same sign as $\gamma$.  
Figure \ref{fig:sim2} shows properties of $DE(\mathcal{T})$ as a function of $\gamma$ under various epidemiologic and study design parameters, when $\beta=0$.  
The top row shows results under block randomization, and the bottom row shows results under cluster randomization.  
The left column shows $DE(\mathcal{T})$ for increasing values of $\sigma^2$, the variability of individual-level susceptibility and infectiousness. 
The middle column shows how $DE(\mathcal{T})$ changes with $\mu_\xi$, the average value of the individual-level infectiousness coefficient.  
When these values are large and negative, few infections are transmitted by infected individuals, so the value of $\gamma$ has little effect on $DE(\mathcal{T})$, which stays near zero. When $\mu_\xi$ is large and positive, something similar happens: infected individuals are highly infectious even when $\gamma<0$, and $DE(\mathcal{T})$ is near zero for a wide range of values of $\gamma$. 
When $\mu_\xi$ is near zero, the value of $\gamma$ fully determines the infectiousness of treated individuals, and $DE(\mathcal{T})$ exhibits the largest difference from zero.  
In the right column, we examine the effect of changes and heterogeneity in the follow-up time $\mathcal{T}$, allowing the observation time $\mathcal{T}_i$ to vary between clusters.   In all cases, the magnitude of the direct effect increases with the absolute value of $\gamma$.  While Propositions \ref{prop:bernoulli} - \ref{prop:cluster} give the sign of $DE(\mathcal{T})$ for any combination of parameter values, simulation results show that the magnitude of $DE(\mathcal{T})$ changes substantially depending on the specific study design and epidemiologic characteristics.  In the Supplement we present a simulation study exploring the properties of $DE(\mathcal{T})$ when $\beta\neq 0$.

\begin{figure}[t]
  \centering
	\includegraphics[scale=0.9]{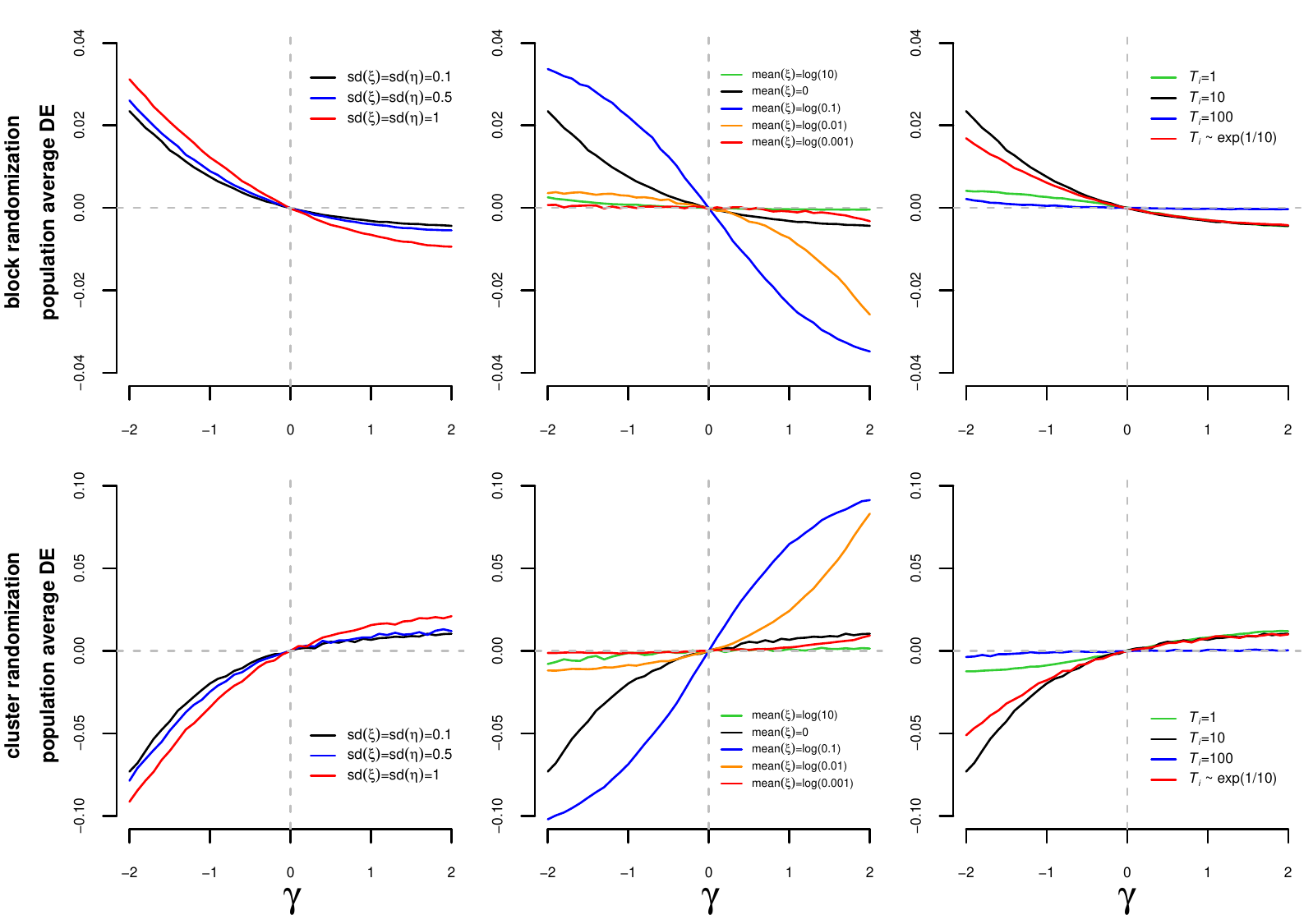} \\
  \caption{Further simulation results for $DE(\mathcal{T})$ under the null hypothesis of no susceptibility effect, $\beta=0$. The top row shows $DE(\mathcal{T})$ under block randomization and the bottom row shows $DE(\mathcal{T})$ under cluster randomization. The left column shows results under different values of the variance of individual-level susceptibility $\eta_{ij}$ and infectiousness $\xi_{ij}$. The middle column shows how $DE(\mathcal{T})$ changes with the mean value of infectiousness $\xi_{ij}$. The right column shows results under different distributions for the cluster-level observation time $\mathcal{T}_i$. }
  \label{fig:sim2}
\end{figure}


\section{Discussion}

\citet{greenwood1915statistics} proposed three conditions for making valid inferences about the effect of a vaccine: 1) ``The persons must be, in all material respects, alike''; 2) ``The effective exposure to the disease must be identical in the case of inoculated and uninoculated persons''; and 3) ``The criteria of the fact of inoculation and of the fact of the disease having occurred must be independent''. Randomization ensures that conditions 1 and 3 are satisfied on average \citep{rothman2008modern,greenland1986identifiability,halloran2010design}.  In this paper, we have shown that under certain randomization designs, 
the direct effect defined by \citet{hudgens2008toward} compares individual infection outcomes in a way that ensures condition 2 
does not hold: treated and untreated subjects experience differential exposure to infectiousness, and $DE_{ij}(t)$ is subject to confounding. 

The direct effect is a well-defined and natural statistical estimand that is identified under randomization with mild assumptions.  But under some randomization designs, it may not provide empirical researchers with the individualistic causal effect they seek: ``the difference betweeen the outcome in the individual with the intervention and what the outcome would have been without the intervention, all other things being equal'' \citep[][page 332]{halloran1991study}, because it does not hold all other things equal.  A heuristic explanation provides useful intuition. 
\begin{enumerate}
  \item Under Bernoulli randomization, treated and untreated subjects are exposed to the \emph{same} number of treated individuals on average. 
  \item Under block randomization, treated subjects are exposed to \emph{fewer} treated individuals ($m_i-1$) than untreated subjects ($m_i$).  
  \item Under cluster randomization, treated subjects are exposed to \emph{more} treated individuals ($n_i-1$) than untreated subjects ($0$). 
  \end{enumerate}
These differences in joint treatment distribution are natural consequences of the randomization designs; Propositions \ref{prop:bernoulli}-\ref{prop:cluster} establish the connection to \emph{differential exposure to infectiousness}, and to the direct effect estimand under the structural transmission model \eqref{eq:haz}. When the null hypothesis of $\beta=0$ is true and an infectiousness effect exists ($\gamma\neq 0$), treated and untreated subjects under block and cluster randomization experience differential exposure to infectiousness that depends on the sign of the infectiousness coefficient $\gamma$.  These results apply to individuals within clusters, and hold for any number of clusters.  Similarly, odds and risk ratios computed by contrasting average individual outcomes under treatment versus no treatment may be subject to the same biases \citep[e.g.][]{morozova2018risk}.

Our main results investigate the direct effect under the null hypothesis $\beta=0$ because this case is analytically tractable, and because preservation of the null is a desirable property of any effect measure or test statistic. Some real-world interventions may have this feature; for example, transmission-blocking vaccines \citep{kaslow2002transmission,delrieu2015design} have negligible susceptibility effect, but may be effective in reducing infectiousness of infected individuals.  Isolation policies may also confer minimal susceptibility benefit to individuals assigned to ``quarantine upon infection'', and a strong beneficial infectiousness effect on their contacts \citep{aiello2016design}.  
For untested interventions like new vaccines, investigators may not know whether the susceptibility effect is beneficial, harmful, or null. 
The results outlined here may apply in cases where the true susceptibility effect $\beta$ is nonzero: when the average infection outcome $\overline{Y}_{ij}(t,\x_i)$ is a continuous function of $\beta$, there may exist an interval around $\beta\neq 0$ in which the direct effect is \emph{biased across the null hypothesis of no susceptibility effect} under some designs. 
Therefore estimating $DE(t)=0$ under block randomization need not imply that the susceptibility effect is null, nor does estimating $DE(t)\neq 0$ imply that the susceptibility effect is not null. 
In particular, simulation results show that under block randomization, a vaccine that both helps prevent infection in each person who receives it ($\beta<0$) and helps prevent transmission upon infection ($\gamma<0$) can nevertheless exhibit $DE(t)>0$.  When $DE(t)$ is interpreted as a causal parameter, investigators may conclude that an effective intervention is harmful to the individuals who receive it because its ``direct effect'' is positive.  Simulation results in the Supplement explore the conditions leading to sign mismatch between $DE(t)$ and $\beta\neq 0$.

In this paper, we employ a relatively simple structural transmission model \eqref{eq:haz} because it is widely used and well understood by infectious disease epidemiologists, the hazard of infection has a simple functional form, and its parameters $\beta$ and $\gamma$ correspond naturally to the susceptibility and infectiousness effects defined by \citet{halloran1997study}.  However, this transmission model does not incorporate additional realistic features of infectious disease transmission, such as a latent infection period, multiple infection, removal/recovery, or treatment following infection. Similarly, we have not modeled heterogeneous contact patterns within clusters, nor violated stratified interference by permitting transmission between clusters.  We conjecture that more sophisticated structural models of infectious disease transmission would not differ in their qualitative implications: dependent randomization designs induce differential exposure to infectiousness whenever the treatment affects infectiousness, resulting in counfounding of the direct effect as a measure of the susceptibility effect. If the direct effect under dependent randomization designs does not provide a meaningful approximation to the susceptibility effect of interest under a simplistic transmission model such as \eqref{eq:haz}, we do not expect it to do so under a richer class of more complex structural transmission models. 


Researchers who wish to avoid the pathologies of the direct effect in a randomized trial have three basic options. First, Proposition \ref{prop:bernoulli} shows that changing the randomization design to Bernoulli allocation within clusters breaks the dependence between $x_{ij}$ and $\x_{i(j)}$ \citep{savje2017average}.  Then the conditional probability $\Pr(\X_{i(j)}=\x_{i(j)}|x_{ij}=x)$ in \eqref{eq:iapo} becomes the marginal probability $\Pr(\X_{i(j)}=\x_{i(j)})$, and the direct effect becomes a simple average of individualistic effects. Second, researchers may target a marginal estimand that does not condition on the assigned treatment, as \citet{vanderweele2011effect} and \citet{savje2017average} recommend.  This approach would permit use of a dependent randomization design by changing the conditional marginalizing distribution in \eqref{eq:iapo} to the unconditional distribution of the treatment to other units, $\Pr(\X_{i(j)}=\x_{i(j)})$, provided this probability is positive under the design.  Third, when structural assumptions are warranted and enough data are available, researchers may choose to fit a structural model similar to \eqref{eq:haz} to estimate parameters (e.g. $\beta$) coresponding to the causal effects of interest \citep{rhodes1996counting,auranen2000transmission,cauchemez2006investigating,kenah2015semi}. 

Finally, we have focused here on three idealized randomization designs that are employed in real-world intervention trials.  Non-randomized (i.e. pragmatic, or observational) studies of interventions or risk factors for infection in clusters occupy an uncertain middle ground.  Even when the intervention or covariate of interest is unrelated to other baseline confounders and independent of the potential infection outcomes, it may be unreasonable to assume that it is distributed independently at random within clusters, as it would be under Bernoulli randomization.  Likewise, strict negative or positive correlation in covariate values, of the kind induced by block and cluster randomization respectively, seems implausible.  When any dependence exists in the distribution of treatment in an observational study, regression adjustment or stratification on baseline covariates may not be sufficient to ensure exchangeability of subjects with respect to infection exposure during the study.  Depending on the distribution of treatment, the relationship between the direction or sign of marginal contrasts and the true susceptibility effect may be difficult to predict.

\if0\blind
{

\noindent \textbf{Acknowledgements}: 
This work was supported by NIH grants NICHD DP2 HD091799-01 and NIDA R36 DA042643. 
We are grateful to
Peter M. Aronow, 
Xiaoxuan Cai, 
Ted Cohen, 
Soheil Eshghi, 
Gregg S. Gonsalves, 
M. Elizabeth Halloran, 
Michael Hudgens, 
Eben Kenah, 
Zehang Li, 
Wen Wei Loh, 
Sida Peng, 
Fredrik S\"avje, 
Yushuf Sharker, 
and
Daniel Weinberger 
for helpful comments.  

} \fi


\section{Proof of Proposition \ref{prop:cluster}}

The proof of Proposition \ref{prop:cluster} proceeds in the same way as the proof of Proposition \ref{prop:block} in the main text.

\begin{proof}[Proof of Proposition~\ref{prop:cluster}]
Define $\x_i^1=(1,\ldots,1)$ and $\x_i^0=(0,\ldots,0)$.
  First, we deduce a stochastic order relation for 
  treatment assigments $\x_i^1$ and $\x_i^0$.
When $\gamma=0$, $F_l(w)=G_l(w)$ for all $l$ and all $w$.  Therefore 
$\tilde T_{il}^1=\tilde T_{il}^0$ for all $l$ and so 
$\tilde Y_{ij}^1(t)=\tilde Y_{ij}^0(t)$ for all $t$.
Then $Y_{ij}(t,\x_i^1)$ is equal in distribution to $Y_{ij}(t,\x_i^0)$ for all $t$  
and so $\overline{Y}_{ij}(t,\x_i^1)=\overline{Y}_{ij}(t,\x_i^0)$.  
When $\gamma<0$, note that $\tilde Y_{ij}^1(t) \le \tilde Y_{ij}^0(t)$ for all $t$ 
if and only if $\tilde T_{ij}^1 \ge \tilde T_{ij}^0$.  
Suppose without loss of generality that subjects are relabeled in order of 
their infection in the constructed process, so the $l$th infection occurs in 
subject $l$, $v_l=l$. The waiting time from infection of subject $l-1$ to 
infection of $l$ has distribution function
\begin{equation}
\begin{split}
F_l(w) &= 1-\exp\left\{ -w \sum_{a\in S_l} e^{\eta_{ia}} 
  \left( \alpha + \sum_{b\in I_l} e^{\gamma x_{ib}^1 + \xi_{ib}}\right) \right\} \\
       &< 1-\exp\left\{ -w \sum_{a\in S_l} e^{\eta_{ia}} 
         \left( \alpha + \sum_{b\in I_l} e^{\xi_{ib}}\right) \right\} \\
       &= 1-\exp\left\{ -w \sum_{a\in S_l} e^{\eta_{ia}} 
         \left( \alpha + \sum_{b\in I_l} e^{\gamma x_{ib}^0 + \xi_{ib}} \right) \right\} \\
       &= G_l(w)
\end{split}
\end{equation}
for all $w$. Therefore $F^{-1}_l(U_l) > G^{-1}_l(U_l)$ by monotonicity of 
$F_l(w)$ and $G_l(w)$, so the constructed infection times are 
\begin{equation}
 \tilde T_{ij}^1 = \sum_{l=1}^j \tilde W_l^1 
              = \sum_{l=1}^j F_l^{-1}(U_l) 
              > \sum_{l=1}^j G_l^{-1}(U_l) 
              = \sum_{l=1}^j \tilde W_l^0 
              = \tilde T_{ij}^0 
\end{equation}
where we interpret an empty sum to be equal to zero.  Therefore 
$\tilde T_{ij}^1 \ge \tilde T_{ij}^0$ and hence 
$\Pr(\tilde Y_{ij}^1(t) \le \tilde Y_{ij}^0(t))=1$.  
By Lemma \ref{lem:dominance}, $Y_{ij}(t,\x^0)$ strictly stochastically 
dominates $Y_{ij}(t,\x^1)$ for all $t>0$.  It follows that the expected 
values of the potential infection outcomes obey 
$\overline{Y}_{ij}(t,\x^1)<\overline{Y}_{ij}(t,\x^0)$ for all $t>0$.
Under cluster randomization, 
  \begin{equation}
    \begin{split}
      DE_{ij}(\mathcal{T}) &= \sum_{\z\in\mathcal{X}^{n_i-1}} 
        \overline{Y}_{ij}(\mathcal{T},1,\z) \Pr(\X_{i(j)}=\z|x_{ij}=1) 
          - \overline{Y}_{ij}(\mathcal{T},0,\z) \Pr(\X_{i(j)}=\z|x_{ij}=0) \\
      &= \sum_{\z\in\mathcal{X}^{n_i-1}} \overline{Y}_{ij}(\mathcal{T},1,\z) \indicator{|\z|=n_i-1} 
        - \overline{Y}_{ij}(\mathcal{T},0,\z)  \indicator{|\z|=0} \\
      &= \overline{Y}_{ij}(\mathcal{T},\x_i^1) - \overline{Y}_{ij}(\mathcal{T},\x_i^0),
    \end{split}
  \end{equation}
where $\mathcal{X}^{n_i}$ be the set of all binary $n_i$-vectors.
Therefore, $DE_{ij}(\mathcal{T})$ can be expressed as a contrast between the 
outcome of $j$ when all subjects are treated, versus when no subjects are 
treated, and we see that $DE_{ij}(\mathcal{T}) < 0$ when $\gamma<0$.
The case of $\gamma>0$ is the same as for $\gamma<0$, 
with inequalities switched. 
\end{proof}


\section{Simulation study}

\subsection{Additional simulation details}

The hazard of infection takes the form given in \eqref{eq:haz} of the main text, where $\beta=0$ and $\gamma$ takes a specified value. 
Unless otherwise noted, the exogenous force of infection is $\alpha=0.01$, the individual susceptibility coefficients $\eta_{ij}$ are independent $\text{Normal}(\mu_\eta,\sigma^2_\eta)$ and infectiousness coefficients $\xi_{ij}$ are independent $\text{Normal}(\mu_\xi,\sigma^2_\xi)$, all individuals were assumed uninfected at baseline $Y_{ij}(0)=0$, the cluster size $n_i$ is $2+\text{Poisson}(2)$, and the observation time is $\mathcal{T}=10$.  Table \ref{tbl:sim.par} summarizes the values of all simulation parameters.

\begin{table}[ht]
\caption{Summary of simulation parameters}
\label{tbl:sim.par} \centering%
\resizebox{0.9\textwidth}{!}{%
\begin{tabular}{l l l}
{} & {} & {} \\ 
\toprule
{Notation} & {Parameter} & {Value} \\
\midrule 
{$\beta$} & {susceptibility effect of $x$} & {0 in Figures \ref{fig:sim1} and \ref{fig:sim2} of the main text} \\
{} & {} & {[-2 ; 2] in Figures \ref{fig:2dsim1} - \ref{fig:2dsim9}} \\[0.75em]
{$\gamma$} & {infectiousness effect of $x$} & {[-2 ; 2]} \\[0.75em]
{$\Delta \beta$, $\Delta \gamma$} & {increment size for $\beta$ and $\gamma$} & {0.1} \\[0.75em]
{$\alpha$} & {external force of infection} & {0.01} \\[0.75em]
{$\eta$} & {individual-level susceptibility} & {unless otherwise noted, $\eta_{ij} \sim N (0, 0.1^2)$} \\[0.75em]
{$\xi$} & {individual-level infectiousness} & {unless otherwise noted, $\xi_{ij} \sim N (0, 0.1^2)$} \\[0.75em]
{$n_i$} & {size of cluster $i$} & {unless otherwise noted, $n_i \sim 2+\text{Pois}(2)$} \\[0.75em]
{$\mathcal{T}$} & {observation time} & {unless otherwise noted, $\mathcal{T} = 10$} \\[0.75em]
{$Y(0)$} & {infections at $t=0$} & {$Y_{ij}(0) = 0$, $j = 1,\ldots, n_i$; $i = 1, \ldots, N$} \\[0.75em]
{$p$} & {treatment assignment probability under} & {0.5} \\
{} & {Bernoulli and cluster randomization} & {} \\[0.75em]
{$m_i$} & {number treated per cluster} & {$\lfloor {n_i/2} \rfloor $}\\
{} & {under block randomization} & {} \\[0.75em]
{$N$} & {number of clusters} & {1000} \\[0.75em]
{$N_s$} & {number of simulations } & {100 - 2000} \\
{} & {per combination of parameter values} & {} \\[0.75em]
\bottomrule
\end{tabular} }
\end{table}

The following estimators are used to compute the population average $DE(\mathcal{T})$ in the simulation study.  Under Bernoulli randomization, define 
\[ \widehat{DE}(\mathcal{T}) = \frac{1}{N}\sum_{i=1}^N \frac{1}{n_i} \sum_{j=1}^{n_i} \frac{y_{ij}x_{ij}}{p} - \frac{y_{ij}(1-x_{ij})}{1-p} .\]
Under block randomization, define 
\[ \widehat{DE}(\mathcal{T}) = \frac{1}{N}\sum_{i=1}^N \frac{\sum_{j=1}^{n_i} y_{ij} x_{ij}}{\sum_{j=1}^{n_i} x_{ij}} - \frac{\sum_{j=1}^{n_i} y_{ij} (1-x_{ij})}{\sum_{j=1}^{n_i} (1-x_{ij})}  .\]
Under cluster randomization, let $S_i=1$ when the cluster is assigned treatment, and let $S_i=0$ otherwise.  Define 
\[ \widehat{DE}(\mathcal{T}) = \frac{\sum_{i=1}^N \frac{S_i}{n_i} \sum_{j=1}^{n_i} y_{ij}}{\sum_{i=1}^N S_i} - \frac{\sum_{i=1}^N \frac{(1-S_i)}{n_i} \sum_{j=1}^{n_i} y_{ij}}{\sum_{i=1}^N (1-S_i)}  .\]

\subsection{Additional simulation results}

Figures \ref{fig:sim1} and \ref{fig:sim2} in the main text illustrate the behavior of the population average $DE(\mathcal{T})$ as a function of the infectiousness effect $\gamma$ under the null hypothesis of $\beta=0$. In this section we provide the results of the simulations for a range of values of the susceptibility effect, $-2 < \beta < 2$.  
In Figures \ref{fig:2dsim1} - \ref{fig:2dsim9}, the top row shows a heat map of the population average $DE(\mathcal{T})$ as a function of the susceptibility effect $\beta$ (horizontal axis) and infectiousness effect $\gamma$ (vertical axis). Blue color corresponds to negative values of $DE(\mathcal{T})$, and red color to positive values. The $DE(\mathcal{T})$ is a direction-unbiased estimate of the susceptibility effect if red color is on the right of the vertical line that corresponds to $\beta=0$, and blue color is on the left of this line.  The bottom row of Figures \ref{fig:2dsim1} - \ref{fig:2dsim9} shows the regions in the two-dimensional $(\beta; \gamma)$ space, where the sign of the $DE(\mathcal{T})$ is opposite that of $\beta$. These regions are colored black.


Figures \ref{fig:2dsim1} - \ref{fig:2dsim3} correspond to the same study designs as those used to produce Figure \ref{fig:sim1} in the paper. 
The $DE(\mathcal{T})$ is direction-unbiased under Bernoulli randomization, while under block and cluster randomization the $DE(\mathcal{T})$ exhibits direction bias in some regions of the $(\beta; \gamma)$ space. Under block randomization, the sign of population average $DE(\mathcal{T})$ is opposite that of $\beta$ when $\beta$ and $\gamma$ have the same sign, and $\gamma$ is more extreme than $\beta$. Under cluster randomization, direction bias of the $DE(\mathcal{T})$ appears in the regions, where $\beta$ and $\gamma$ have opposite signs. Figure \ref{fig:2dsim3} shows that under cluster randomization and a given set of simulation parameters, when $\beta < 0$, the region of direction bias is very small. Absence of black regions in the upper left quadrants of the bottom row plots in Figure \ref{fig:2dsim3} is an artifact of the chosen range of values of $\beta$, as well as the step size.
The region of direction bias gets smaller with the increase of the cluster size. \\

\begin{figure} 
\centering
	\includegraphics[width=0.95\textwidth]{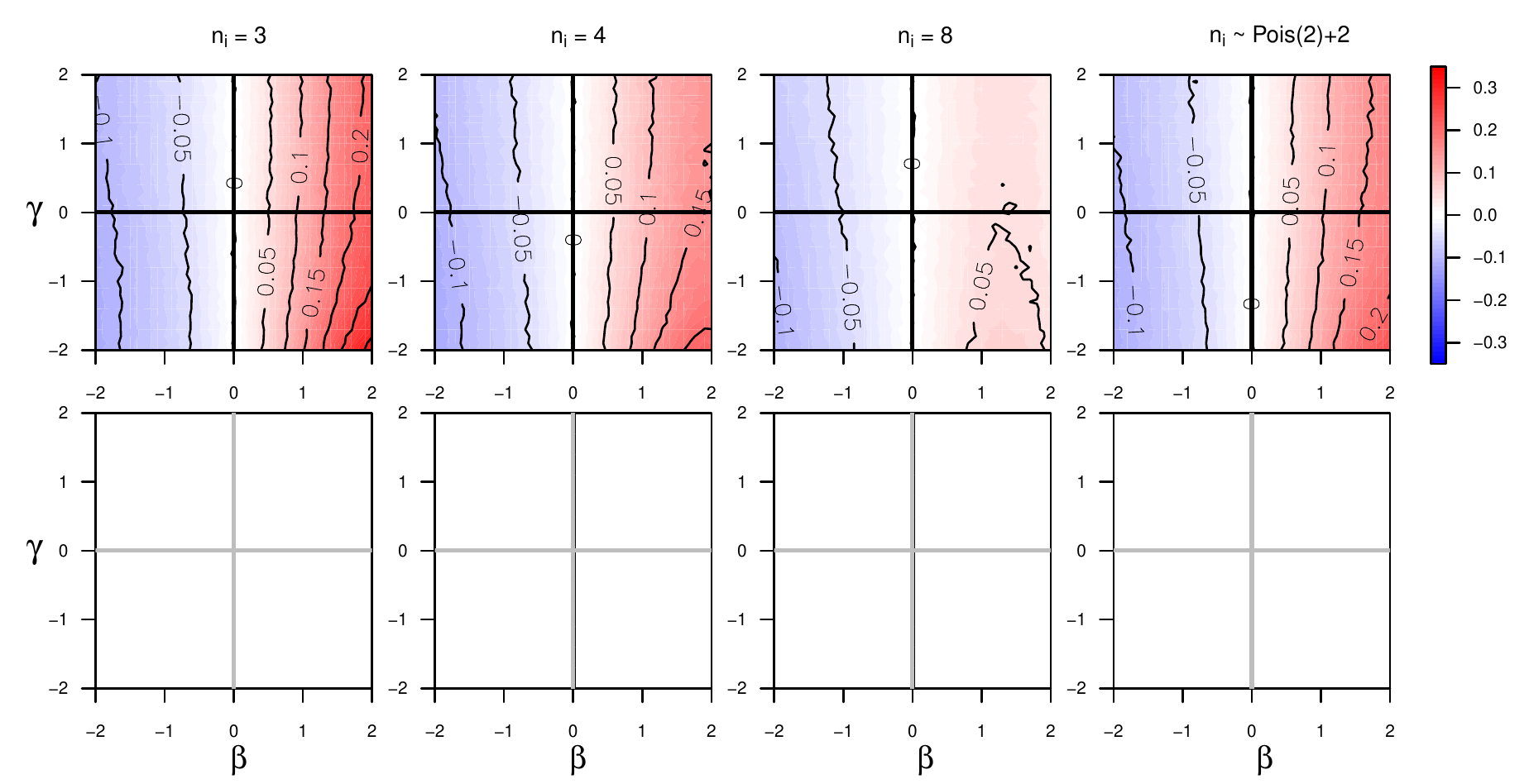} 
	\caption{Population average DE under Bernoulli randomization and different cluster sizes.}
	\label{fig:2dsim1}
\end{figure}

\begin{figure} 
\centering
	\includegraphics[width=0.95\textwidth]{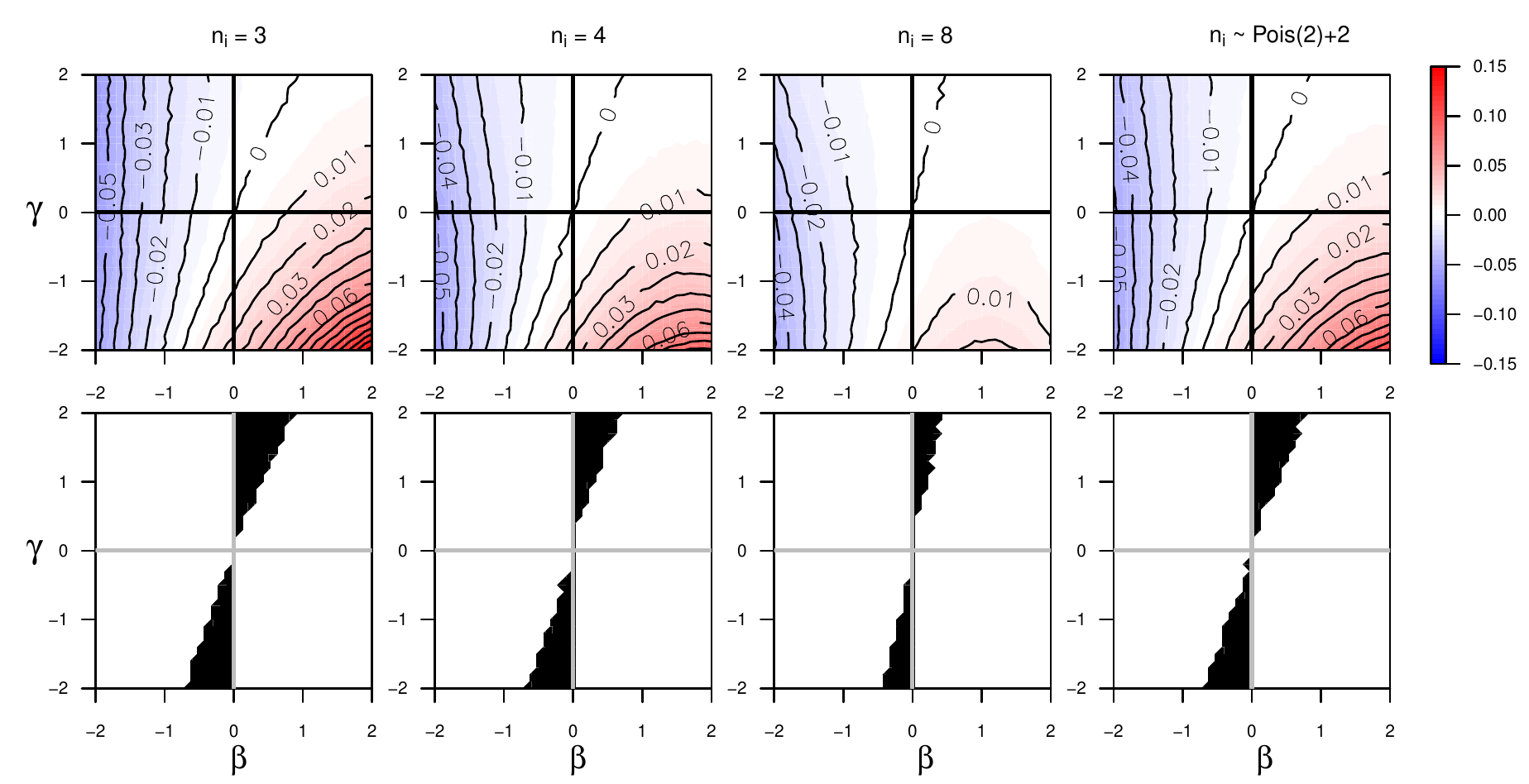} 
	\caption{Population average DE under block randomization and different cluster sizes.}
	\label{fig:2dsim2}
\end{figure}

\begin{figure} 
\centering
	\includegraphics[width=0.95\textwidth]{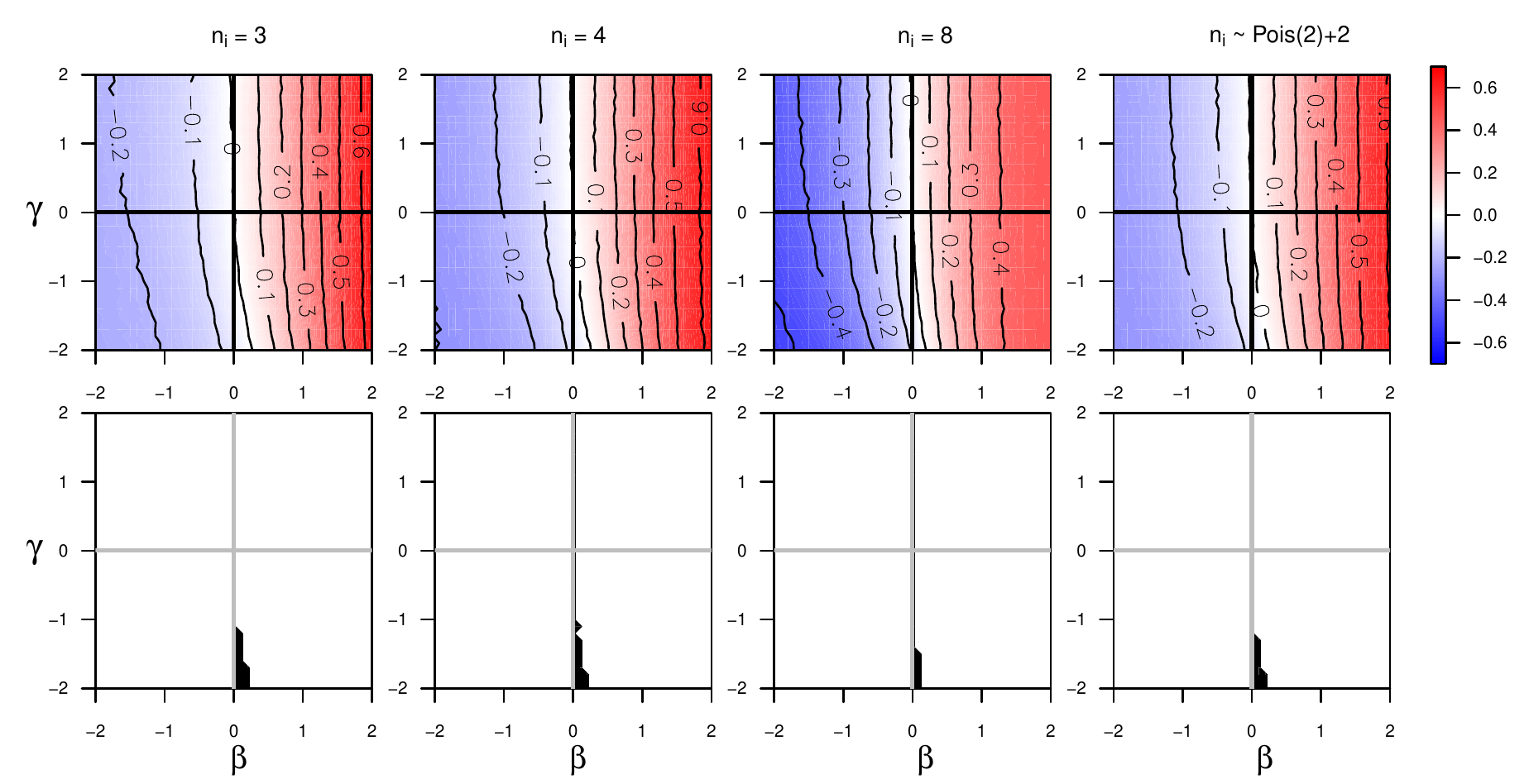} 
	\caption{Population average DE under cluster randomization and different cluster sizes.}
	\label{fig:2dsim3}
\end{figure}

The magnitude of $DE(\mathcal{T})$ under the null of $\beta=0$ is not necessarily related to the size of the direction-bias region when $\beta \ne 0$.  Figure \ref{fig:sim1} in the main text shows that under the null, cluster randomization results in a larger size of the bias compared to block randomization. At the same time, the region in the $(\beta; \gamma)$ space where the $DE(\mathcal{T})$ exhibits direction bias is larger under block compared to cluster randomization (all other thing being equal). This happens because under cluster randomization the $DE(\mathcal{T})$ changes substantially more rapidly in response to one unit change in the value of $\beta$ compared to the $DE(\mathcal{T})$ under block randomization.

\begin{figure} 
\centering
	\includegraphics[width=0.7\textwidth]{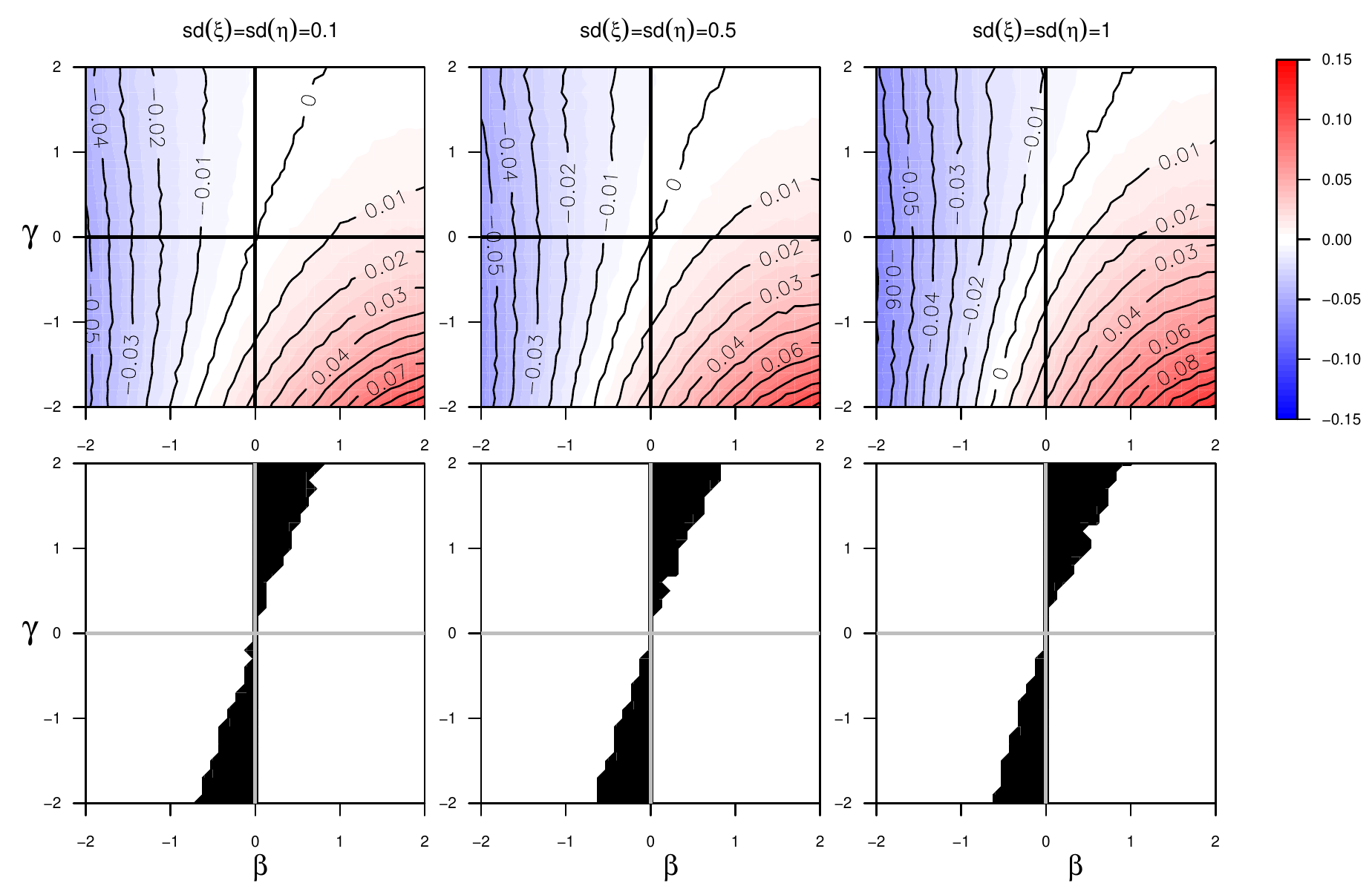} 
	\caption{Population average DE under block randomization and different variance of individual-level susceptibility ($\eta$) and infectiousness ($\xi$).}
	\label{fig:2dsim4}
\end{figure}

\begin{figure} 
\centering
	\includegraphics[width=0.7\textwidth]{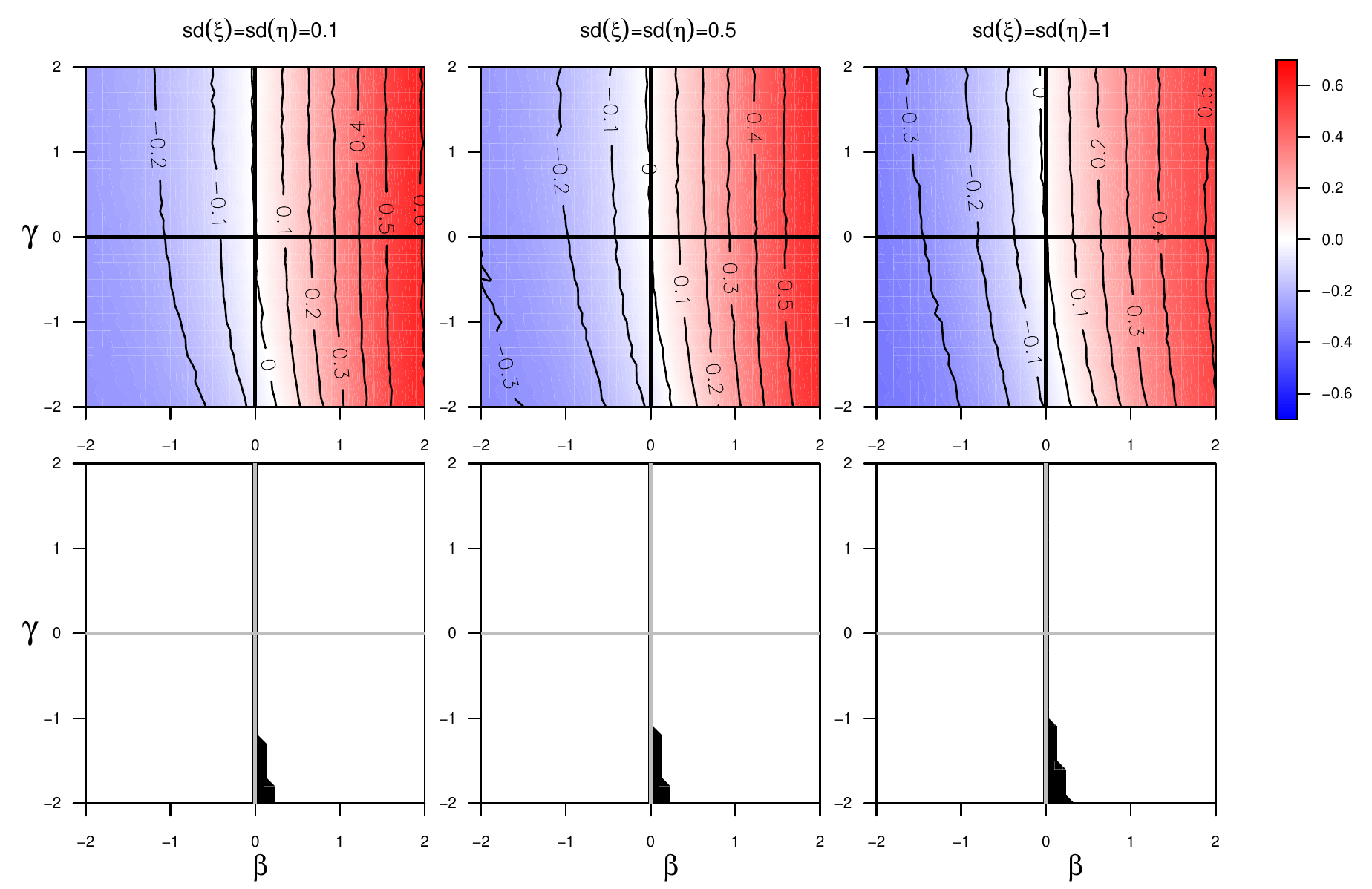} 
	\caption{Population average DE under cluster randomization and different variance of individual-level susceptibility ($\eta$) and infectiousness ($\xi$).}
	\label{fig:2dsim5}
\end{figure}

Figures \ref{fig:2dsim4} - \ref{fig:2dsim5} correspond to the same study designs as the left column of Figure \ref{fig:sim2} in the main text. 
The region of direction bias increases with the increase of variance of untreated individual-level susceptibility ($\eta$) and infectiousness ($\xi$).

\begin{figure} 
\centering
	\includegraphics[width=\textwidth]{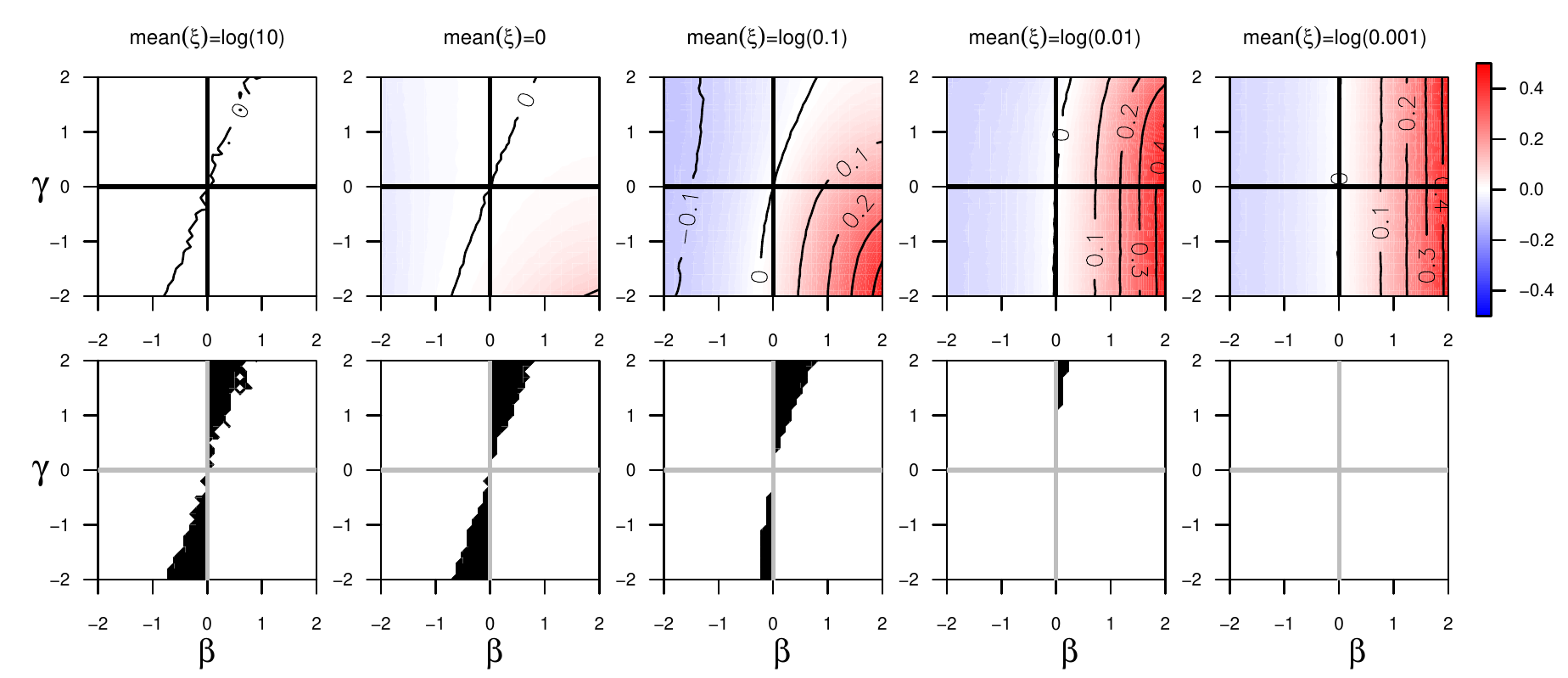} 
	\caption{Population average DE under block randomization and different average individual-level infectiousness ($\xi$).}
	\label{fig:2dsim6}
\end{figure}

\begin{figure} 
\centering
	\includegraphics[width=\textwidth]{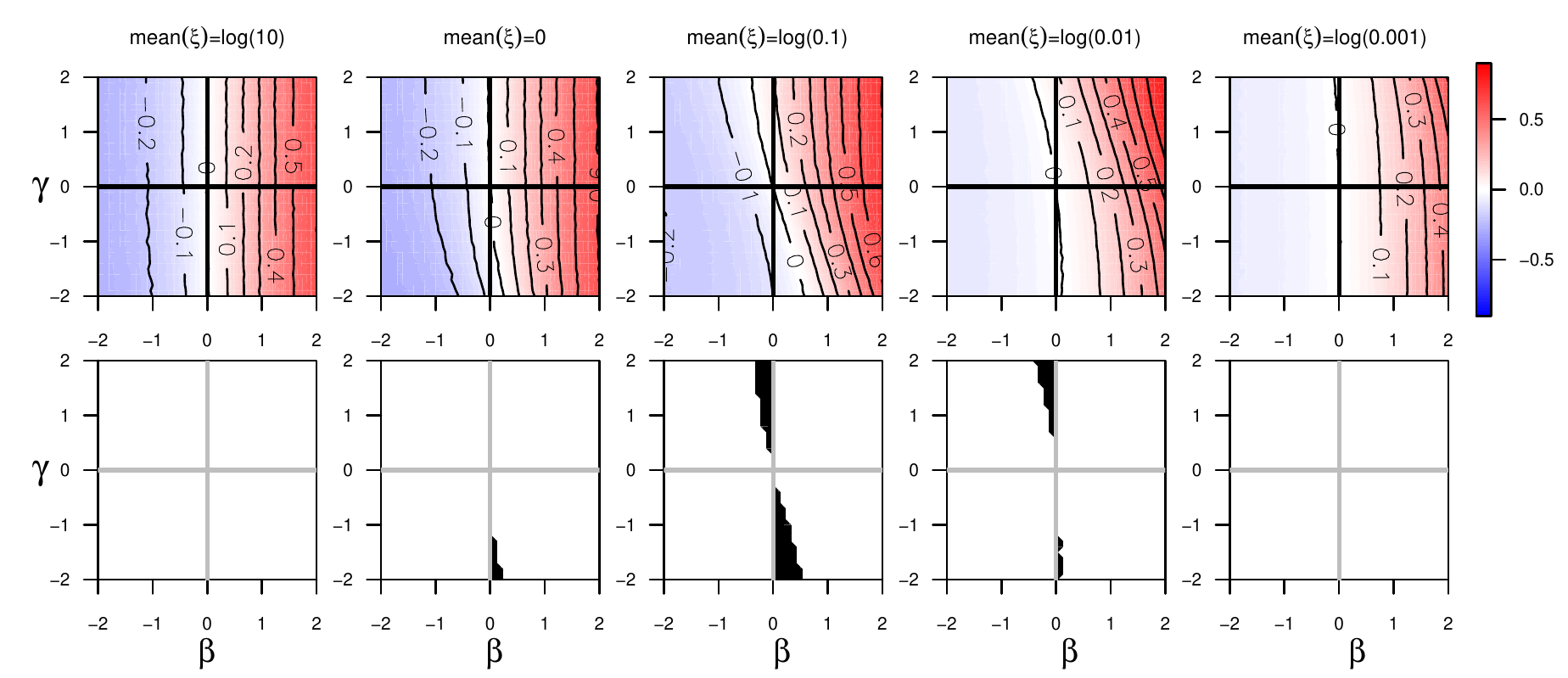} 
	\caption{Population average DE under cluster randomization and different average individual-level infectiousness ($\xi$).}
	\label{fig:2dsim7}
\end{figure}

Figures \ref{fig:2dsim6} - \ref{fig:2dsim7} correspond to the same study designs as the middle column of Figure \ref{fig:sim2} in the main text.  Under the block randomization the region of direction bias gets smaller as the mean untreated within-cluster infectiousness decreases (Figure \ref{fig:2dsim6}). However, under cluster randomization this relationship is non-monotonic: the region of direction bias is very small for extreme (small or large) values of average untreated within-cluster infectiousness, and largest when the mean of $\xi$ is somewhere in the middle (Figure \ref{fig:2dsim7}).

\begin{figure}[H]
\centering
	\includegraphics[width=0.95\textwidth]{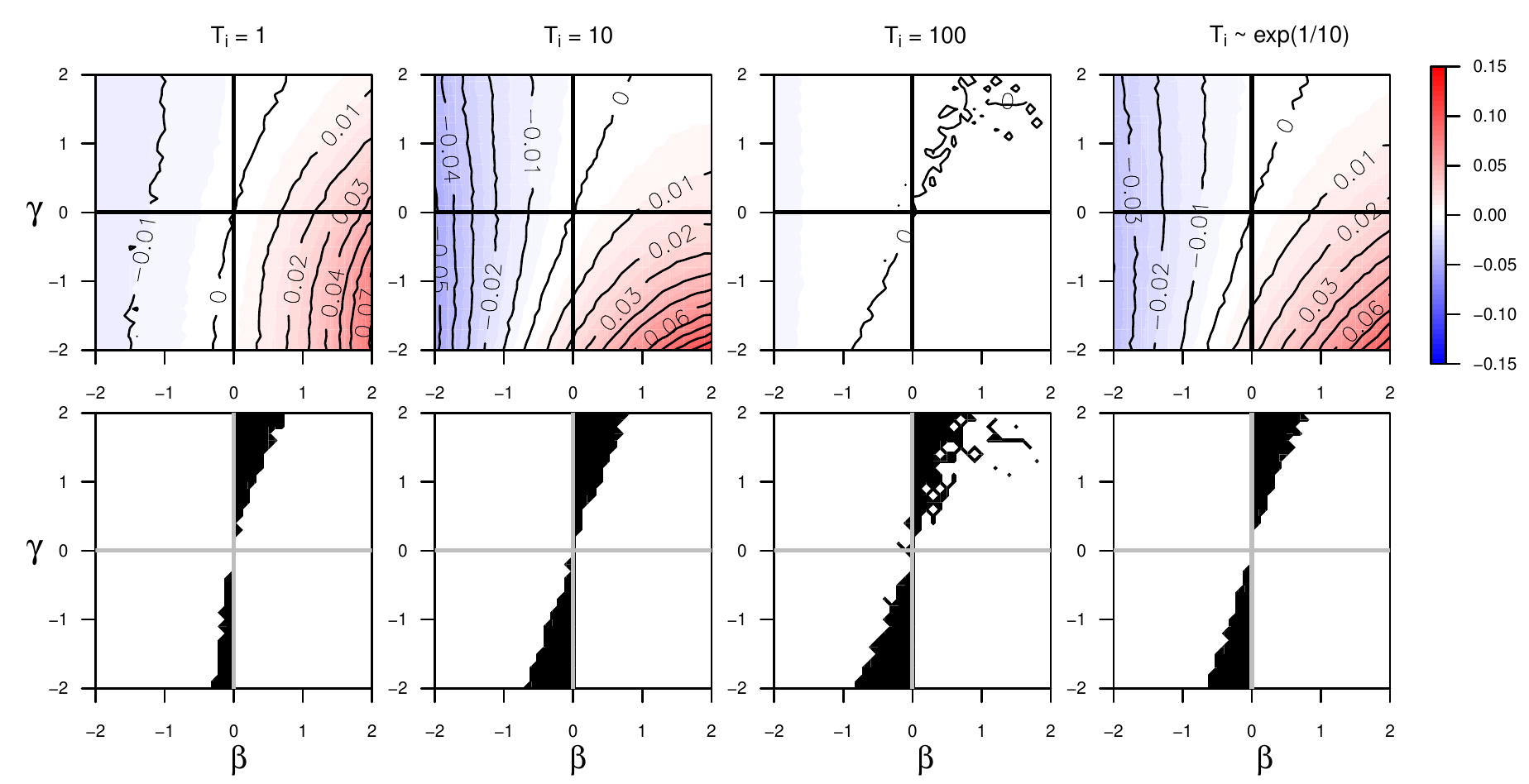} 
	\caption{Population average DE under block randomization and different observation time.}
	\label{fig:2dsim8}
\end{figure}

\begin{figure}[H]
\centering
	\includegraphics[width=0.95\textwidth]{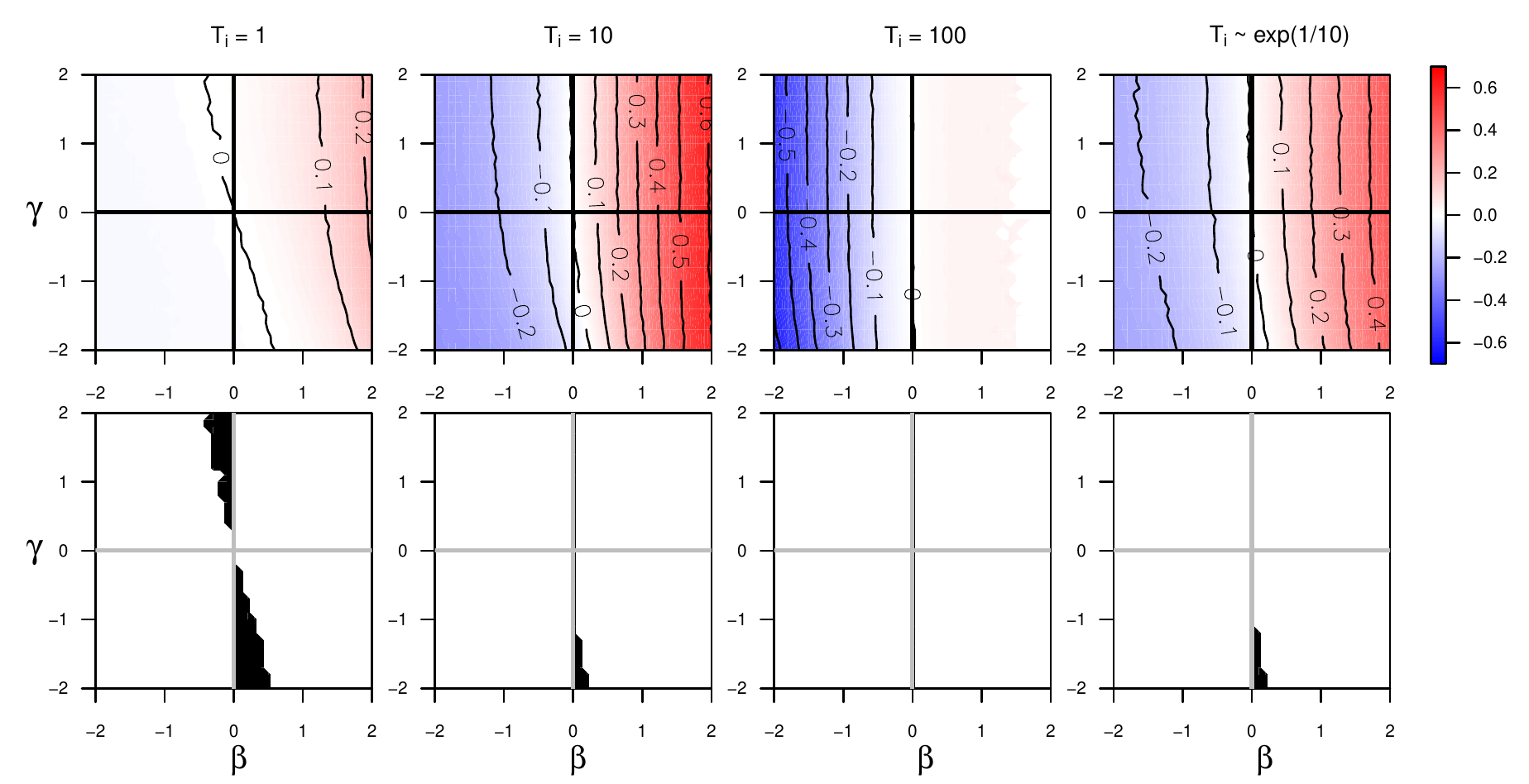} 
	\caption{Population average DE under cluster randomization and different observation time.}
	\label{fig:2dsim9}
\end{figure}

Figures \ref{fig:2dsim8} - \ref{fig:2dsim9} correspond to the same study designs as the right column of Figure \ref{fig:sim2} in the main text.
The region where the $DE(\mathcal{T})$ exhibits direction bias as an estimate of the susceptibility effect increases with the observation time under block randomization (Figure \ref{fig:2dsim8}), but decreases under cluster randomization (Figure \ref{fig:2dsim9}).



\begin{table}[H]
\caption{Parameters that don't change across simulations}
\label{sim.par.same} \centering%
\begin{tabular}{l l}
{} & {} \\ 
\toprule
{Parameter} & {Value} \\
\midrule
{$\beta$} & {0} \\
{$\gamma$} & {[-2 ; 2]} \\
{Step for $\gamma$} & {0.1} \\
{Number of clusters, $N$} & {$N=1000$} \\
{Number of simulations per value of $\gamma$, $N_s$} & {$N_s = 500$} \\
{External FOI, $\alpha_i$} & {$\alpha_i = 0.01$, $i = 1, \ldots, N$} \\
{Infections at $t=0$} & {$y_{ij}(0) = 0$, $j = 1,\ldots, n_i$; $i = 1, \ldots, N$} \\
\bottomrule
\end{tabular} 
\end{table}

\begin{table}[H]
\caption{Parameters that vary across simulations}
\label{sim.par.dif} \centering%
\resizebox{\textwidth}{!}{%
\begin{tabular}{l l l l l l l}
{} & {} & {} & {} & {} \\ 
\toprule
{\#} & {Randomization}   & {Cluster size,} & {Tx assignment} & {Distribution} & {Distribution} & {$\mathcal{T}_i$}\\
{  } &  {design}           & {$n_i$}     & {parameter}  & {of $\eta_{ij}$} & {of $\xi_{ij}$} & {} \\
\midrule
{1.a} & {Bernoulli} & {$n_i = 3, \forall i$} & {$p=0.5$} & {$\eta_{ij} \sim N (\mu=0, \sigma=0.1)$} & {$\xi_{ij} \sim N (\mu=0, \sigma=0.1)$} & {$\mathcal{T}_i=10$} \\ 
{1.b} & {Bernoulli} & {$n_i = 4, \forall i$} & {$p=0.5$} & {$\eta_{ij} \sim N (\mu=0, \sigma=0.1)$} & {$\xi_{ij} \sim N (\mu=0, \sigma=0.1)$} & {$\mathcal{T}_i=10$} \\ 
{1.c} & {Bernoulli} & {$n_i = 8, \forall i$} & {$p=0.5$} & {$\eta_{ij} \sim N (\mu=0, \sigma=0.1)$} & {$\xi_{ij} \sim N (\mu=0, \sigma=0.1)$} & {$\mathcal{T}_i=10$} \\ 
{1.d} & {Bernoulli} & {$n_i \sim 2+\text{Pois}(2)$} & {$p=0.5$} & {$\eta_{ij} \sim N (\mu=0, \sigma=0.1)$} & {$\xi_{ij} \sim N (\mu=0, \sigma=0.1)$} & {$\mathcal{T}_i=10$} \\[0.5em] \hline
%
%
{2.a} & {Block} & {$n_i = 3, \forall i$} & {$m_i=1, \forall i$} & {$\eta_{ij} \sim N (\mu=0, \sigma=0.1)$} & {$\xi_{ij} \sim N (\mu=0, \sigma=0.1)$} & {$\mathcal{T}_i=10$} \\
{2.b} & {Block} & {$n_i = 4, \forall i$} & {$m_i=2, \forall i$} & {$\eta_{ij} \sim N (\mu=0, \sigma=0.1)$} & {$\xi_{ij} \sim N (\mu=0, \sigma=0.1)$} & {$\mathcal{T}_i=10$} \\ 
{2.c} & {Block} & {$n_i = 8, \forall i$} & {$m_i=4, \forall i$} & {$\eta_{ij} \sim N (\mu=0, \sigma=0.1)$} & {$\xi_{ij} \sim N (\mu=0, \sigma=0.1)$} & {$\mathcal{T}_i=10$} \\ 
{2.d} & {Block} & {$n_i \sim 2+\text{Pois}(2)$} & {$m_i=\lfloor {n_i/2} \rfloor, \forall i$} & {$\eta_{ij} \sim N (\mu=0, \sigma=0.1)$} & {$\xi_{ij} \sim N (\mu=0, \sigma=0.1)$} & {$\mathcal{T}_i=10$} \\[0.5em] \hline
%
{3.a} & {Cluster} & {$n_i = 3, \forall i$} & {$p=0.5$} & {$\eta_{ij} \sim N (\mu=0, \sigma=0.1)$} & {$\xi_{ij} \sim N (\mu=0, \sigma=0.1)$} & {$\mathcal{T}_i=10$} \\
{3.b} & {Cluster} & {$n_i = 4, \forall i$} & {$p=0.5$} & {$\eta_{ij} \sim N (\mu=0, \sigma=0.1)$} & {$\xi_{ij} \sim N (\mu=0, \sigma=0.1)$} & {$\mathcal{T}_i=10$} \\ 
{3.c} & {Cluster} & {$n_i = 8, \forall i$} & {$p=0.5$} & {$\eta_{ij} \sim N (\mu=0, \sigma=0.1)$} & {$\xi_{ij} \sim N (\mu=0, \sigma=0.1)$} & {$\mathcal{T}_i=10$} \\
{3.d} & {Cluster} & {$n_i \sim 2+\text{Pois}(2)$} & {$p=0.5$} & {$\eta_{ij} \sim N (\mu=0, \sigma=0.1)$} & {$\xi_{ij} \sim N (\mu=0, \sigma=0.1)$} & {$\mathcal{T}_i=10$} \\[0.5em] \hline
%
{4.a} & {Block} & {$n_i \sim 2+\text{Pois}(2)$} & {$m_i=\lfloor {n_i/2} \rfloor, \forall i$} & {$\eta_{ij} \sim N (\mu=0, \sigma=0.01)$} & {$\xi_{ij} \sim N (\mu=0, \sigma=0.01)$} & {$\mathcal{T}_i=10$} \\
{4.b} & {Block} & {$n_i \sim 2+\text{Pois}(2)$} & {$m_i=\lfloor {n_i/2} \rfloor, \forall i$} & {$\eta_{ij} \sim N (\mu=0, \sigma=0.1)$} & {$\xi_{ij} \sim N (\mu=0, \sigma=0.1)$} & {$\mathcal{T}_i=10$} \\
{4.c} & {Block} & {$n_i \sim 2+\text{Pois}(2)$} & {$m_i=\lfloor {n_i/2} \rfloor, \forall i$} & {$\eta_{ij} \sim N (\mu=0, \sigma=0.5)$} & {$\xi_{ij} \sim N (\mu=0, \sigma=0.5)$} & {$\mathcal{T}_i=10$} \\
{4.d} & {Block} & {$n_i \sim 2+\text{Pois}(2)$} & {$m_i=\lfloor {n_i/2} \rfloor, \forall i$} & {$\eta_{ij} \sim N (\mu=0, \sigma=1)$} & {$\xi_{ij} \sim N (\mu=0, \sigma=1)$} & {$\mathcal{T}_i=10$} \\[0.5em]  \hline
{5.a} & {Cluster} & {$n_i \sim 2+\text{Pois}(2)$} & {$p=0.5$} & {$\eta_{ij} \sim N (\mu=0, \sigma=0.01)$} & {$\xi_{ij} \sim N (\mu=0, \sigma=0.01)$} & {$\mathcal{T}_i=10$} \\
{5.b} & {Cluster} & {$n_i \sim 2+\text{Pois}(2)$} & {$p=0.5$} & {$\eta_{ij} \sim N (\mu=0, \sigma=0.1)$} & {$\xi_{ij} \sim N (\mu=0, \sigma=0.1)$} & {$\mathcal{T}_i=10$} \\
{5.c} & {Cluster} & {$n_i \sim 2+\text{Pois}(2)$} & {$p=0.5$} & {$\eta_{ij} \sim N (\mu=0, \sigma=0.5)$} & {$\xi_{ij} \sim N (\mu=0, \sigma=0.5)$} & {$\mathcal{T}_i=10$} \\
{5.d} & {Cluster} & {$n_i \sim 2+\text{Pois}(2)$} & {$p=0.5$} & {$\eta_{ij} \sim N (\mu=0, \sigma=1)$} & {$\xi_{ij} \sim N (\mu=0, \sigma=1)$} & {$\mathcal{T}_i=10$} \\[0.5em] \hline
{6.a} & {Block} & {$n_i \sim 2+\text{Pois}(2)$} & {$m_i=\lfloor {n_i/2} \rfloor, \forall i$} & {$\eta_{ij} \sim N (\mu=0, \sigma=0.1)$} & {$\xi_{ij} \sim N (\mu=\log(10), \sigma=0.1)$} & {$\mathcal{T}_i=10$} \\
{6.b} & {Block} & {$n_i \sim 2+\text{Pois}(2)$} & {$m_i=\lfloor {n_i/2} \rfloor, \forall i$} & {$\eta_{ij} \sim N (\mu=0, \sigma=0.1)$} & {$\xi_{ij} \sim N (\mu=\log(1), \sigma=0.1)$} & {$\mathcal{T}_i=10$} \\
{6.c} & {Block} & {$n_i \sim 2+\text{Pois}(2)$} & {$m_i=\lfloor {n_i/2} \rfloor, \forall i$} & {$\eta_{ij} \sim N (\mu=0, \sigma=0.1)$} & {$\xi_{ij} \sim N (\mu=\log(0.1), \sigma=0.1)$} & {$\mathcal{T}_i=10$} \\
{6.d} & {Block} & {$n_i \sim 2+\text{Pois}(2)$} & {$m_i=\lfloor {n_i/2} \rfloor, \forall i$} & {$\eta_{ij} \sim N (\mu=0, \sigma=0.1)$} & {$\xi_{ij} \sim N (\mu=\log(0.01), \sigma=0.1)$} & {$\mathcal{T}_i=10$} \\
{6.e} & {Block} & {$n_i \sim 2+\text{Pois}(2)$} & {$m_i=\lfloor {n_i/2} \rfloor, \forall i$} & {$\eta_{ij} \sim N (\mu=0, \sigma=0.1)$} & {$\xi_{ij} \sim N (\mu=\log(0.001), \sigma=0.1)$} & {$\mathcal{T}_i=10$} \\[0.5em]  \hline
{7.a} & {Cluster} & {$n_i \sim 2+\text{Pois}(2)$} & {$p=0.5$} & {$\eta_{ij} \sim N (\mu=0, \sigma=0.1)$} & {$\xi_{ij} \sim N (\mu=\log(10), \sigma=0.1)$} & {$\mathcal{T}_i=10$} \\
{7.b} & {Cluster} & {$n_i \sim 2+\text{Pois}(2)$} & {$p=0.5$} & {$\eta_{ij} \sim N (\mu=0, \sigma=0.1)$} & {$\xi_{ij} \sim N (\mu=\log(1), \sigma=0.1)$} & {$\mathcal{T}_i=10$} \\
{7.c} & {Cluster} & {$n_i \sim 2+\text{Pois}(2)$} & {$p=0.5$} & {$\eta_{ij} \sim N (\mu=0, \sigma=0.1)$} & {$\xi_{ij} \sim N (\mu=\log(0.1), \sigma=0.1)$} & {$\mathcal{T}_i=10$} \\
{7.d} & {Cluster} & {$n_i \sim 2+\text{Pois}(2)$} & {$p=0.5$} & {$\eta_{ij} \sim N (\mu=0, \sigma=0.1)$} & {$\xi_{ij} \sim N (\mu=\log(0.01), \sigma=0.1)$} & {$\mathcal{T}_i=10$} \\
{7.e} & {Cluster} & {$n_i \sim 2+\text{Pois}(2)$} & {$p=0.5$} & {$\eta_{ij} \sim N (\mu=0, \sigma=0.1)$} & {$\xi_{ij} \sim N (\mu=\log(0.001), \sigma=0.1)$} & {$\mathcal{T}_i=10$} \\[0.5em]  \hline
{8.a} & {Block} & {$n_i \sim 2+\text{Pois}(2)$} & {$m_i=\lfloor {n_i/2} \rfloor, \forall i$} & {$\eta_{ij} \sim N (\mu=0, \sigma=0.1)$} & {$\xi_{ij} \sim N (\mu=0, \sigma=0.1)$} & {$\mathcal{T}_i=1$} \\
{8.b} & {Block} & {$n_i \sim 2+\text{Pois}(2)$} & {$m_i=\lfloor {n_i/2} \rfloor, \forall i$} & {$\eta_{ij} \sim N (\mu=0, \sigma=0.1)$} & {$\xi_{ij} \sim N (\mu=0, \sigma=0.1)$} & {$\mathcal{T}_i=10$} \\
{8.c} & {Block} & {$n_i \sim 2+\text{Pois}(2)$} & {$m_i=\lfloor {n_i/2} \rfloor, \forall i$} & {$\eta_{ij} \sim N (\mu=0, \sigma=0.1)$} & {$\xi_{ij} \sim N (\mu=0, \sigma=0.1)$} & {$\mathcal{T}_i=100$} \\
{8.d} & {Block} & {$n_i \sim 2+\text{Pois}(2)$} & {$m_i=\lfloor {n_i/2} \rfloor, \forall i$} & {$\eta_{ij} \sim N (\mu=0, \sigma=0.1)$} & {$\xi_{ij} \sim N (\mu=0, \sigma=0.1)$} & {$\mathcal{T}_i \sim \exp(1/10)$} \\[0.5em]  \hline
{9.a} & {Cluster} & {$n_i \sim 2+\text{Pois}(2)$} & {$p=0.5$} & {$\eta_{ij} \sim N (\mu=0, \sigma=0.1)$} & {$\xi_{ij} \sim N (\mu=0, \sigma=0.1)$} & {$\mathcal{T}_i=1$} \\
{9.b} & {Cluster} & {$n_i \sim 2+\text{Pois}(2)$} & {$p=0.5$} & {$\eta_{ij} \sim N (\mu=0, \sigma=0.1)$} & {$\xi_{ij} \sim N (\mu=0, \sigma=0.1)$} & {$\mathcal{T}_i=10$} \\
{9.c} & {Cluster} & {$n_i \sim 2+\text{Pois}(2)$} & {$p=0.5$} & {$\eta_{ij} \sim N (\mu=0, \sigma=0.1)$} & {$\xi_{ij} \sim N (\mu=0, \sigma=0.1)$} & {$\mathcal{T}_i=100$} \\
{9.d} & {Cluster} & {$n_i \sim 2+\text{Pois}(2)$} & {$p=0.5$} & {$\eta_{ij} \sim N (\mu=0, \sigma=0.1)$} & {$\xi_{ij} \sim N (\mu=0, \sigma=0.1)$} & {$\mathcal{T}_i \sim \exp(1/10)$} \\[0.5em]
\bottomrule
\end{tabular} }
\end{table}

\bibliographystyle{plainnat}
\bibliography{epireg}

\end{document}